%% file: visibilityGraph.tex
\documentclass[a4paper,10pt]{article}
\usepackage[utf8]{inputenc}
\usepackage{fullpage}
\usepackage{url}
\usepackage{graphicx}
\usepackage{amssymb}
\usepackage{amsmath}
\usepackage{amsthm}
\usepackage{epsfig}
\usepackage{color}
\usepackage{algorithm}
\usepackage[noend]{algpseudocode}
\usepackage{authblk}

\title{A Characterization of Visibility Graphs for Pseudo-Polygons}
\author[1]{Matt Gibson \thanks{gibson@cs.utsa.edu}}
\author[2]{Erik Krohn \thanks{uvg160@my.utsa.edu}}
\author[1]{Qing Wang \thanks{krohne@uwosh.edu}}
\affil[1]{Dept. of Computer Science, University of Texas at San Antonio, San Antonio, TX, USA}
\affil[2]{Dept. of Computer Science, University of Wisconsin - Oshkosh, Oshkosh, WI, USA}

 \newtheorem{theorem}{Theorem}
 
 \newtheorem{lemma}[theorem]{Lemma}

\newcounter{NecCond}
\newtheorem{nc}[NecCond]{Necessary Condition}
\newcommand{\Li}{\mathcal{L}}

\newcommand{\p}{\partial}

\begin{document}
\bibliographystyle{plain}
\maketitle

\begin{abstract}
In this paper, we give a characterization of the visibility graphs of pseudo-polygons.  We first identify some key combinatorial properties of pseudo-polygons, and we then give a set of five necessary conditions based off our identified properties.  We then prove that these necessary conditions are also sufficient via a reduction to a characterization of vertex-edge visibility graphs given by O'Rourke and Streinu.
\end{abstract}

\section{Introduction}
Geometric covering problems have been a focus of research for decades.  Here we are given some set of points $P$ and a set $S$ where each $s \in S$ can cover some subsets of $P$.  The subset of $P$ is generally induced by some geometric object.  For example, $P$ might be a set of points in the plane, and $s$ consists of the points contained within some disk in the plane.  For most variants, the problem is NP-hard and can easily be reduced to an instance of the combinatorial set cover problem which has a polynomial-time $O(\log n)$-approximation algorithm, which is the best possible approximation under standard complexity assumptions \cite{Feige:2003}.  The main question therefore is to determine for which variants of 
geometric set cover can we obtain polynomial-time approximation algorithms with approximation ratio $o(\log n)$, as any such algorithm must exploit the geometry of the problem to achieve the result.  This area has been studied extensively, see for example \cite{Aronov:2010,Varadarajan09,Aloupis:2009}, and much progress has been made utilizing algorithms that are based on solving the standard linear programming relaxation.  

Unfortunately this technique has severe limitations for some variants of geometric set cover, and new ideas are needed to make progress on these variants.  In particular, the techniques are lacking when the points $P$ we wish to cover is a simple polygon, and we wish to place the smallest number of points in $P$ that collectively ``see'' the polygon.  This problem is classically referred to as the \textit{art gallery problem} as an art gallery can be modeled as a polygon and the points placed by an algorithm represent cameras that can ``guard'' the art gallery.  This has been one of the most well-known problems in computational geometry for many years, yet still to this date the best polynomial-time approximation algorithm for this problem is a $O(\log n)$-approximation algorithm.  The key issue is a fundamental lack of understanding of the combinatorial structure of visibility inside simple polygons.  It seems that in order to develop powerful approximation algorithms for this problem, the community first needs to better understand the underlying structure of such visibility.  

\noindent\textbf{Visibility Graphs.}  A very closely related issue which has received a lot of attention in the community is the \textit{visibility graph} of a simple polygon.  Given a simple polygon $P$, the visibility graph $G = (V,E)$ of $P$ has the following structure.  For each vertex $p \in P$, there is a vertex in $V$, and there is an edge connecting two vertices in $G$ if and only if the corresponding vertices in $P$ ``see'' each other (i.e., the line segment connecting the points does not go outside the polygon).  Two major open problems regarding visibility graphs of simple polygons are the visibility graph characterization problem and the visibility graph recognition problem.  The \textit{visibility graph characterization} problem seeks to define a set of properties that all visibility graphs satisfy.  The \textit{visibility graph recognition} problem is the following.  Given a graph $G$, determine if there exists a simple polygon $P$ such that $G$ is the visibility graph of $P$ in polynomial time.  

The problems of characterizing and recognizing the visibility graphs of simple polygons have had partial results given dating back to over 25 years ago \cite{Ghosh88} and remain open to this day with only a few special cases being solved.  Characterization and recognition results have been given in the special cases of ``spiral'' polygons \cite{EverettC95} and ``tower polygons'' \cite{ChoiSC95}.  There have been several results \cite{Ghosh97,EverettC95,SrinivasaraghavanM94} that collectively have led to four necessary conditions that a simple polygon visibility graph must satisfy.  That is, if the graph $G$ does not satisfy all four of the conditions then we know that $G$ is not the visibility graph for \textit{any} simple polygon, and moreover it can be determined if a graph $G$ satisfies all of the necessary conditions in polynomial time.  Streinu, however, has given an example of graph that satisfies all of the necessary conditions but is not a visibility graph for any simple polygon \cite{Streinu05}, implying that the set of conditions is not sufficient and therefore a strengthening of the necessary conditions is needed.  Unfortunately it is not even known if simple polygon visibility graph recognition is in NP.  See \cite{GhoshG13} for a nice survey on these problems and other related visibility problems.

%Note that a visibility graph $G$ of a simple polygon $P$ must contain a Hamiltonian cycle because a vertex must see the vertices on either side of it (otherwise the polygon would be self-intersecting and thus not simple).  Since determining if a graph contains a Hamiltonian cycle is NP-hard, previous research has assumed that $G$ does has such a cycle $C$ and the vertices are labeled in counterclockwise order according to this cycle.  

\noindent\textbf{Pseudo-polygons.}  Given the difficulty of understanding simple polygon visibility graphs, O'Rourke and Streinu \cite{ORourkeS97} considered the visibility graphs for a special case of polygons called \textit{pseudo-polygons} which we will now define.  An arrangement of \textit{pseudo-lines} $\Li$ is a collection of simple curves, each of which separates the plane, such that each pair of pseudo-lines of $\Li$ intersects at exactly one point, where they cross.  Let $P = \{p_0, p_2,\ldots, p_{n-1}\}$ be a set of points in $\mathbb{R}^2$, and let $\Li$ be an arrangement of $\binom{n}{2}$ pseudo-lines such that every pair of points $p_i$ and $p_j$ lie on exactly one pseudo-line in $\Li$, and each pseudo-line in $\Li$ contains exactly two points of $P$. The pair $(P, \Li)$ is called a \textit{pseudo configuration of points} (pcp) in general position.

Intuitively a pseudo-polygon is determined similarly to a standard Euclidean simple polygon except using pseudo-lines instead of straight line segments.  Let $L_{i,j}$ denote the pseudo-line through the points $p_i$ and $p_j$.  We view $L_{i,j}$ as having three different components.  The subsegment of $L_{i,j}$ connecting $p_i$ and $p_j$ is called the \textit{segment}, and we denote it $p_ip_j$. Removing $p_ip_j$ from $L_{i,j}$ leaves two disjoint \textit{rays}.  Let $r_{i,j}$ denote the ray starting from $p_i$ and moving away from $p_j$, and we let $r_{j,i}$ denote the ray starting at $p_j$ and moving away from $p_i$.  Consider the pseudo line $L_{i,i+1}$ in a pcp (indices taken modulo $n$ and are increasing in counterclockwise order throughout the paper).  We let $e_i$ denote the segment of this line.  A \textit{pseudo-polygon} is obtained by taking the segments $e_i$ for $i \in \{0,\ldots, n-1\}$ if (1) the intersection of $e_i$ and $e_{i+1}$ is only the point $p_{i+1}$ for all $i$, and (2) distinct segments $e_i$ and $e_j$ do not intersect for all $j \neq i+1$. We call the segments $e_i$ the \textit{boundary edges}.  A pseudo-polygon separates the plane into two regions: ``inside'' the pseudo-polygon and ``outside'' the pseudo-polygon, and any two points $p_i$ and $p_j$ see each other if the segment of their pseudo-line does not go outside of the pseudo-polygon.  See Fig \ref{fig:pPolygon} for an illustration.  Pseudo-polygons can be viewed as a combinatorial abstraction of simple polygons.  Note that every simple polygon is a pseudo-polygon (simply allow each $L_{i,j}$ to be the straight line through $p_i$ and $p_j$), and Streinu showed that there are pseudo-polygons that cannot be ``stretched'' into a simple polygon \cite{Streinu05}.

\begin{figure}
\centering
\begin{tabular}{c@{\hspace{0.1\linewidth}}c}
\includegraphics[scale=0.4]{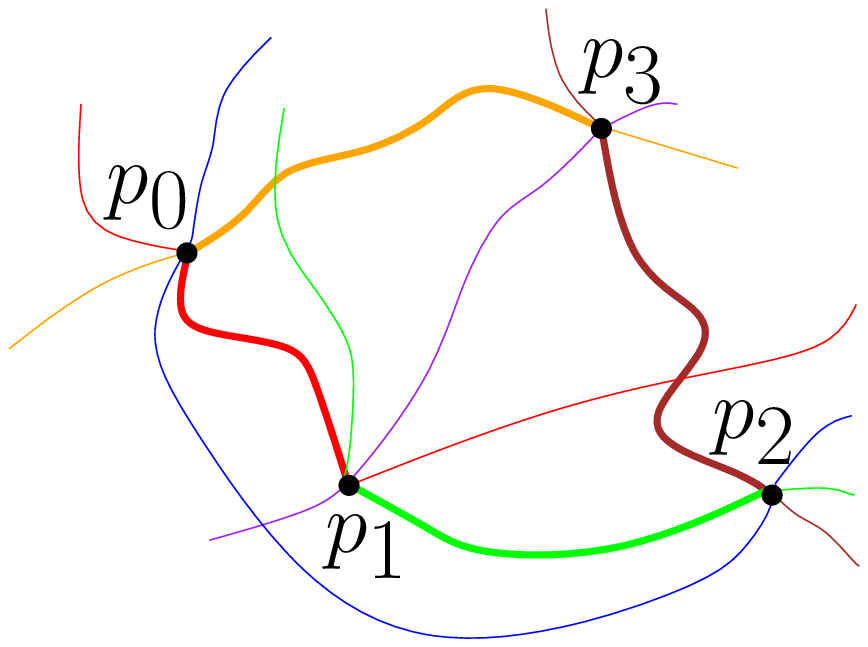}&
\includegraphics[scale=0.4]{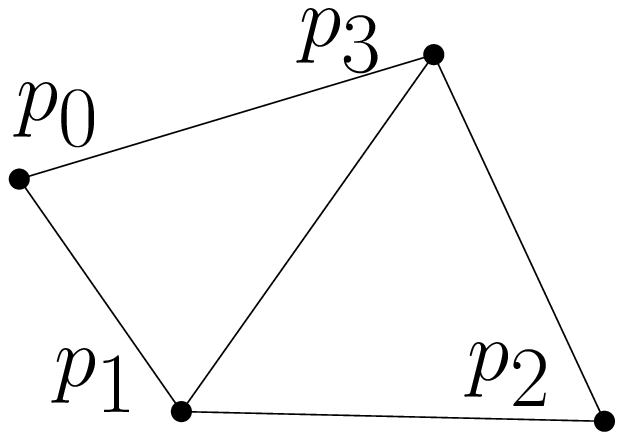}\\
(a) & (b) 
\end{tabular}
\caption{(a) A pcp and pseudo-polygon.  (b) The corresponding visibility graph.  }
\label{fig:pPolygon}
\end{figure}

O'Rourke and Streinu \cite{ORourkeS97} give a characterization of \textit{vertex-edge} visibility graphs of pseudo-polygons.  In this setting, for any vertex $v$ we are told which \textit{edges} $v$ sees rather than which vertices it sees.  Unfortunately, O'Rourke and Streinu showed that vertex-edge visibility graphs encode more information about a pseudo-polygon than a regular visibility graph \cite{ORourke1998105}, and the regular visibility graph characterization problem has remained open for over fifteen years.

\noindent\textbf{Our Results.}  In this paper, we give a characterization of the visibility graphs of pseudo-polygons.  We first identify some key combinatorial properties of pseudo-polygons, and we then give a set of five necessary conditions based off our identified properties.  We then prove that these necessary conditions are also sufficient via a reduction to O'Rourke and Streinu's vertex-edge characterization \cite{ORourkeS97}.  That is, for any visibility graph $G$ that satisfies all necessary conditions, we construct a vertex-edge visibility graph $G_{VE}$ that corresponds with $G$ and show that it satisfies the characterization properties.  Since all simple polygons are pseudo-polygons, our necessary conditions also apply to simple polygon visibility graphs, and in some cases extend or generalize the previously given necessary conditions given for simple polygon visibility graphs \cite{GhoshG13}.  Each of the four necessary conditions given for simple polygons \cite{GhoshG13} have been proved using geometric arguments, yet each of them are implied by the necessary conditions we give for pseudo-polygons which are proved without geometric arguments. Given that not all pseudo-polygons are simple polygons \cite{Streinu05}, additional necessary conditions will be needed to characterize the visibility graphs of simple polygons. 

\section{Preliminaries}
 We begin with some preliminaries and definitions that will be relied upon heavily in our proof.  Our main focus of this paper is to determine if a graph $G$ is the visibility graph for some pseudo-polygon.  Note that the visibility graph $G$ of a pseudo-polygon $P$ must contain a Hamiltonian cycle because each $p_i$ must see $p_{i-1}$ and $p_{i+1}$.  Since determining if a graph contains a Hamiltonian cycle is NP-hard, previous research has assumed that $G$ does have such a cycle $C$ and the vertices are labeled in counterclockwise order according to this cycle.  So now suppose we are given an arbitrary graph $G = (V, E)$ with the vertices labeled $p_0$ to $p_{n-1}$ such that $G$ contains a Hamiltonian cycle $C = (p_0, p_2, \ldots, p_{n-1})$ in order according to their indices.  We are interested in determining if $G$ is the visibility graph for some pseudo-polygon $P$ where $C$ corresponds with the boundary of $P$.  For any two vertices $p_i$ and $p_j$, we let $\p(p_i,p_j)$ denote the vertices and boundary edges encountered when walking counterclockwise around $C$ from $p_i$ to $p_j$ (inclusive).  For any edge $\{p_i,p_j\}$ in $G$, we say that $\{p_i,p_j\}$ is a \textit{visible pair}, as their points in $P$ must see one another.  If $\{p_i,p_j\}$ is not an edge in $G$, then we call $(p_i,p_j)$ and $(p_j,p_i)$ \textit{invisible pairs}.  Note that visible pairs are unordered, and invisible pairs are ordered (for reasons described below).    

Consider any invisible pair $(p_i,p_j)$.  If $G$ is the visibility graph for a pseudo-polygon $P$, the segment of $L_{i,j}$ must exit $P$.  For example, suppose we want to construct a polygon $P$ such that the graph in Fig \ref{fig:blockEx} (a) is the visibility graph of $P$.  Note that $p_0$ should not see $p_2$, and thus if there exists such a polygon, it must satisfy that $p_0p_2$ exits the polygon.  In the case of a simple polygon, we view this process as placing the vertices of $P$ in convex position and then contorting the boundary of $P$ to block $p_0$ from seeing $p_2$.  We can choose $p_1$ or $p_3$ to block $p_0$ from seeing $p_2$ (see (b) and (c)).  Note that as in Fig \ref{fig:blockEx} (b) when using $p_1 \in \p(p_0,p_2)$ as the blocker in a simple polygon, the line segment $p_0p_1$ does not go outside $P$ and the ray $r_{1,0}$ first exits $P$ through a boundary edge in $\p(p_2,p_0)$.  Similarly as in Fig \ref{fig:blockEx} (c) when using $p_3 \in \p(p_2,p_0)$ as the blocker, the line segment $p_0p_3$ does not go outside of the polygon and the ray $r_{3,0}$ first exits the polygon through a boundary edge in $\p(p_1,p_3)$.  The situation is similar in the case of pseudo-polygons, but since we do not have to use straight lines to determine visibility, instead of bending the the boundary of $P$ to block the invisible pair we can instead bend the pseudo-line.  See Fig \ref{fig:blockEx} (d) and (e).  Note that the combinatorial structure of the pseudo-line shown in part (d) (resp. part (e)) is the same as the straight line in part (b) (resp. in part (c)).  The following definition plays an important role in our characterization.  Consider a pseudo-polygon $P$, and let $p_i$ and $p_j$ be two vertices of $P$ that do not see each other.  We say a vertex $p_k \in \p(p_i,p_j)$ of $P$ is a \textit{designated blocker} for the invisible pair $(p_i,p_j)$ if $p_i$ sees $p_k$ (i.e. the segment $p_ip_k$ is inside the polygon) and the ray $r_{i,k}$ first exits the polygon through an edge in $\p(p_j,p_i)$.  The definition for $p_k \in \p(p_j,p_i)$ is defined similarly.  See Figure \ref{fig:OnePairTwoBlocker} (a) for an illustration.  Intuitively, a designated blocker is a canonical vertex that prevents the points in an invisible pair from seeing each other.  In this section, we will prove a key structural lemma of pseudo-polygons: every invisible pair in any pseudo-polygon $P$ has exactly one designated blocker.  

\begin{figure}
\centering
\begin{tabular}{c@{\hspace{0.1\linewidth}}c@{\hspace{0.1\linewidth}}c@{\hspace{0.1\linewidth}}c@{\hspace{0.1\linewidth}}c}
\includegraphics[scale=0.45]{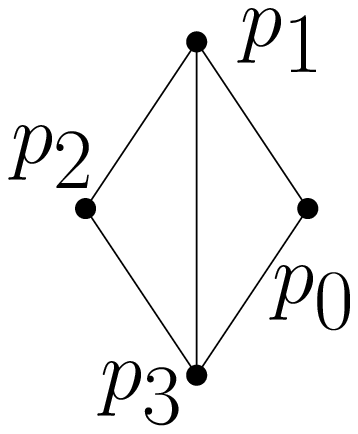}&
\includegraphics[scale=0.45]{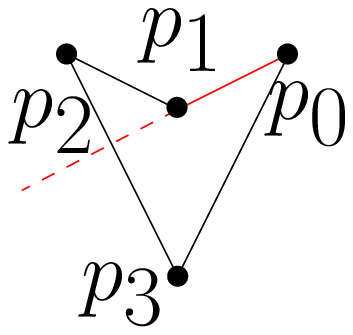}&
\includegraphics[scale=0.45]{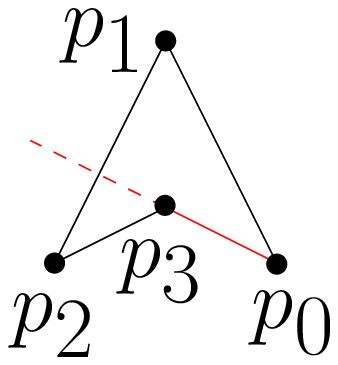}&
\includegraphics[scale=0.45]{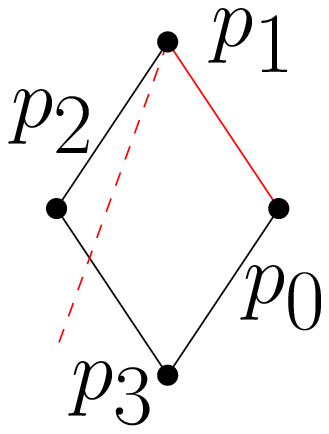}&
\includegraphics[scale=0.45]{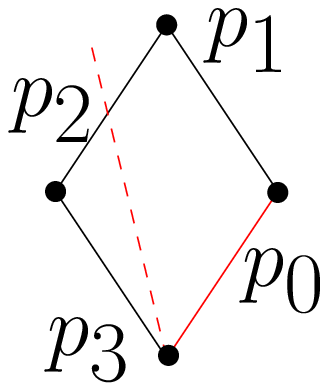}\\
(a) & (b) & (c) & (d) & (e)
\end{tabular}
\caption{(a) A visibility graph $G$.  (b) A simple polygon using $p_1$ to block $p_0$ and $p_2$.  (c) A simple polygon using $p_3$ to block $p_0$ and $p_2$. (d) A pseudo-polygon using $p_1$ to block $p_0$ and $p_2$.  (e) A pseudo-polygon using $p_3$ to block $p_0$ and $p_2$. }
\label{fig:blockEx}
\end{figure}

We now give several definitions and observations that will be used in the proof of the key lemma.  Consider an input graph $G$ with Hamiltonian cycle $C$, and let $(p_i,p_j)$ be an invisible pair in $G$.  If $G$ is the visibility graph of a pseudo-polygon, then there must be some vertex in $G$ that serves as the designated blocker for $(p_i,p_j)$.  The following definition gives a set of at most two candidate vertices for this role.  Starting from $p_j$, walk clockwise towards $p_i$ until we reach the first point $p_k$ such that $\{p_i,p_k\}$ is a visible pair (clearly there must be such a point since $\{p_i, p_{i+1}\}$ is a visible pair).  We say that $p_k$ is a \textit{candidate blocker} for $(p_i, p_j)$ if there are no visible pairs $\{p_s,p_t\}$ such that $p_s \in \p(p_i,p_{k-1})$ and $p_t \in \p(p_{k+1},p_j)$.  Similarly, walk counterclockwise from $p_j$ to $p_i$ until we reach the first point $p_{k'}$ such that $\{p_i,p_{k'}\}$ is a visible pair.  Then $p_{k'}$ is a candidate blocker for $(p_i, p_j)$ if there are no visible pairs $\{p_s,p_t\}$ such that $p_s \in \p(p_j,p_{{k'}-1})$ and $p_t \in \p(p_{{k'}+1},p_i)$.  Note that a vertex may be a candidate blocker for $(p_i, p_j)$ but not for $(p_j, p_i)$.  It clearly follows from the definition that $(p_i,p_j)$ can have at most two candidate blockers: at most one in $\p(p_i,p_j)$ and at most one in $\p(p_j,p_i)$.  We will see that if a vertex in $G$ is not a candidate blocker for $(p_i,p_j)$, then it cannot serve as a designated blocker for $(p_i,p_j)$ in $P$.

\begin{figure}
\centering
\begin{tabular}{c@{\hspace{0.1\linewidth}}c@{\hspace{0.1\linewidth}}c}
\includegraphics[scale=0.35]{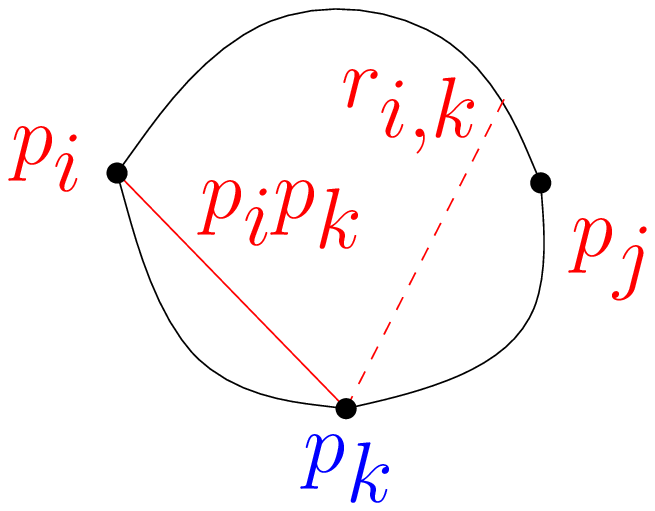}&
\includegraphics[scale=0.35]{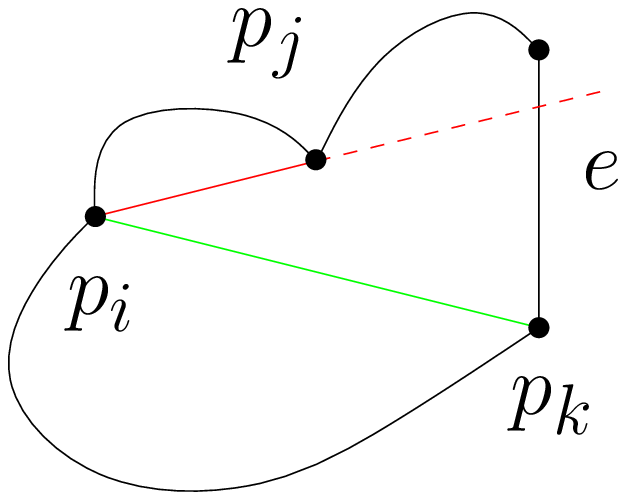}&
\includegraphics[scale=0.35]{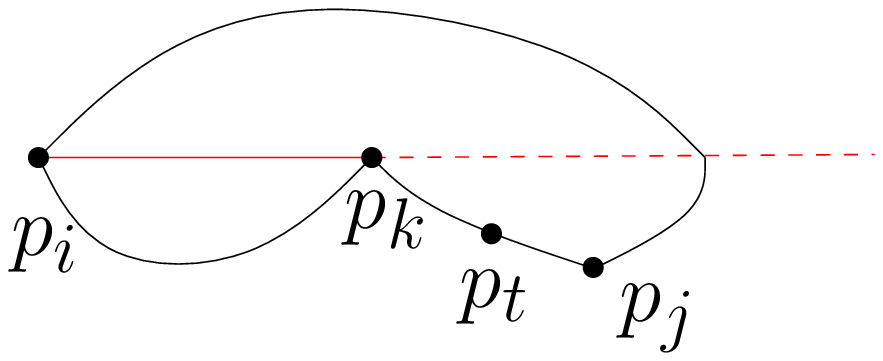}\\
(a) & (b) & (c)
\end{tabular}
\caption{(a) A designated blocker.  (b) The vertex-edge pair $(p_i,e)$ has two witnesses.  Therefore $p_i$ sees $e$. (c) If $p_k$ is the designated blocker for $(p_i,p_j)$ then it also is for $(p_i,p_t)$.}
\label{fig:OnePairTwoBlocker}
\end{figure}

%We now give some intuition for candidate blockers for an invisible pair $(p_i,p_j)$.  In Euclidean polygons, the line segment connecting the corresponding points $p_i$ and $p_j$ must exit the polygon, and therefore there must be a ``blocking vertex'' of the polygon that ``cuts through'' the segment $p_ip_j$.  We will see that there is a similar concept in pseudo-polygons: the segment of the pseudo-line from $p_i$ to $p_j$ must exit the pseudo-polygon and ``wrap around'' a candidate blocker prior to reaching $p_j$.  See Figure \ref{fig:OnePairTwoBlocker} (c) for an illustration.

%%%%%%%%%%%%%%%%%%%%%%%%%%%%%%5%%%

%\textbf{Psuedo-Polygon Structure.}  We now identify the key structure of pseudo-polygons that plays a crucial role in our characterization of their visibility graphs.  

We utilize some observations regarding the vertex-edge visibility graphs for pseudo-polygons given by O'Rourke and Streinu \cite{ORourkeS97} in the proof of our key lemma as well.  We first formally define what it means for a vertex to see a boundary edge in a pseudo-polygon.  Vertex $p_j$ is a \textit{witness} for the vertex-edge pair ($p_i$, $e$) if and only if either 

\begin{enumerate}
\item $p_i$ and $p_j$ are both endpoints of $e$ (permitting $p_j$ = $p_i$), or
\item $p_i$ is not an endpoint of $e$, and both of the following occur: (a) $p_i$ sees $p_j$, and (b) $p_j$ is an endpoint of $e$, or the first boundary edge intersected by $r_{j,i}$ is $e$.
\end{enumerate}

Given the definition of a witness, we say vertex $p$ \textit{sees} edge $e$ if and only if there are at least two witnesses for ($p$, $e$).  See Fig \ref{fig:OnePairTwoBlocker} (b).  The definition requires two witnesses as a vertex $p_i$ could see one endpoint of $e$ without seeing any other part of the edge, and in this situation it is defined that $p_i$ does not see $e$.  We now give the following lemma relating edge visibility and vertex visibility.  Some similar results for straight-line visibility were given in \cite{ORourke1998105}, and we prove them in the context of pseudo-visibility.

\begin{lemma}
If a vertex $p_i$ sees edges $e_{j-1}$ and $e_j$, then it sees vertex $p_j$.  Also if a vertex $p_i$ sees vertex $p_j$, then it sees at least one of $e_{j-1}$ and $e_j$.
\label{lem:edgeVisibility}
\end{lemma}

\begin{proof}
First we will show that if a vertex $p_i$ sees edges $e_{j-1}$ and $e_j$, then it sees vertex $p_j$. Suppose $p_i$ does not see $p_j$, then we have a witness for $e_j$ that intersects $e_j$ and a witness for $e_{j-1}$ that intersects $e_{j-1}$. The line $L_{i,j}$ must stay between the witness lines (because it intersects the witness lines at $p_i$ and therefore cannot intersect them again). In order to block $p_i$ from $p_j$, we'd have to block a witness, a contradiction. See Figure \ref{fig:lemma1} (a). 

Then we will show that if a vertex $p_i$ sees vertex $p_j$, then $p_i$ sees at least one of $e_{j-1}$ and $e_j$. If $p_i$ sees $p_{j-1}$ or $p_{j+1}$ then clearly $p_i$ would see the corresponding edge, so suppose that $p_i$ does not see either of $p_{j-1}$ or $p_{j+1}$.  Recall $L_{i,j}$ partitions the plane into into two half planes.  We consider two cases based on the position of $p_{j-1}$ with respect to these half planes.

\textbf{Case 1:}  First suppose $p_{j-1}$ is to the ``right'' of the ray shot from $p_i$ to $p_j$. See Figure \ref{fig:lemma1} (b). Line $L_{i,j}$ and line $L_{j-1,j}$ divide the plane into four quadrants: one containing $p_i$ and $p_{j-1}$, one containing only $p_{j-1}$, one containing $p_j$ and $p_{j-1}$, and one containing all three points.  If $p_{j+1}$ is in the quadrant containing $p_i$ and $p_{j-1}$, then it follows that any associated pseudo-polygon will have that $p_i$ does not see $p_j$.  Since $p_i$ sees $p_j$, assume $p_{j+1}$ is in any other quadrant, and consider the ``triangle'' $p_ip_jp_{j-1}$. Since $p_i$ cannot see $p_{j-1}$, line $L_{i,j-1}$ should intersect a boundary segment prior to reaching $p_{j_1}$. The boundary segment $e_j$ cannot cut through this triangle given the location of $p_{j+1}$, so the only boundary segments that can cut through $L_{i,j-1}$ are edges that entering the triangle through the segment $L_{i,j-1}$ without intersecting the segment $L_{i,j}$.  It follows that there must be at least one vertex contained inside of the triangle, and we will show that one of these vertices is a witness for $p_i$ and $e_{j-1}$.  For any vertex $p_k$ in the triangle, the line $L_{i,k}$ must stay ``between'' the lines $L_{i,j}$ and $L_{i,j-1}$ and eventually intersect $e_{j-1}$.  Starting at $p_j$, walk towards $p_{j-1}$ along $e_{j-1}$ until we reach the first such exit point of a line $L_{i,k}$.  Let $q$ denote this point.  We claim that all other vertices $p_{k'}$ in the triangle must be below $L_{i,k}$, and therefore $L_{i,k}$ does not intersect any boundary points prior to intersecting $e_{j-1}$. If there were a point $p_{k'}$ in the triangle that is above $L_{i,k}$, then $L_{i,k'}$ intersects $L_{i,k}$ at $p_i$, they split apart, then they must intersect again in the triangle because $L_{i,k'}$ must cross $e_{j-1}$ below $q$.  It follows that $p_k$ is a witness.  See Figure \ref{fig:lemma1} (c).%For the sake of contradiction, suppose that $L_{i,k}$ intersects t%There is a edge that entering this triangle from $p_i$ through $p_{j-1}$ edge and block the lines that start from $p_i$ and exit from $e_{j-1}$, we denote that $p_r$ is the first vertex in clockwise order from $p_j$ to $p_{j-1}$ that a line from $p_i$ and exit from $e_{j-1}$. Line $L_{i,k}$ exits the polygon at 
%$p_r$. If any line exits from $e_{j-1}$, it must in $\p(p_r, p_{j-1})$, because $p_r$ is the first vertex that in $e_{j-1}$. Then there must be another edge that entering this triangle from $p_i$ through $p_{j-1}$ edge and block the lines that start from $p_i$ and exit from $e_{j-1}$, we call this vertex $p_k'$. $L_{i,k}$ and $L_{i,k'}$ intersect at $p_i$ and $p_r$ is the first vertex that in $e_{j-1}$, we have that $L_{i,k}$ and $L_{i,k'}$ have to intersect twice. A contradiction. See Figure \ref{fig:lemma1} (c).

\textbf{Case 2:} Now suppose $p_{j-1}$ is in the left of line $L_{i,j}$ half plane. See Figure \ref{fig:lemma1} (d). Again the lines $L_{i,j}$ and $L_{j,j-1}$ divide the plane into four quadrants.  It must be that $p_{i+1}$ is in the quadrant containing all three points, as otherwise $p_i$ will not be able to see $p_j$.  The analysis here is similar to the previous case, except here we the triangle $p_ip_jp_{j+1}$. 
 \end{proof}

\begin{figure}
\centering
\begin{tabular}{c@{\hspace{0.1\linewidth}}c}
\includegraphics[scale=0.4]{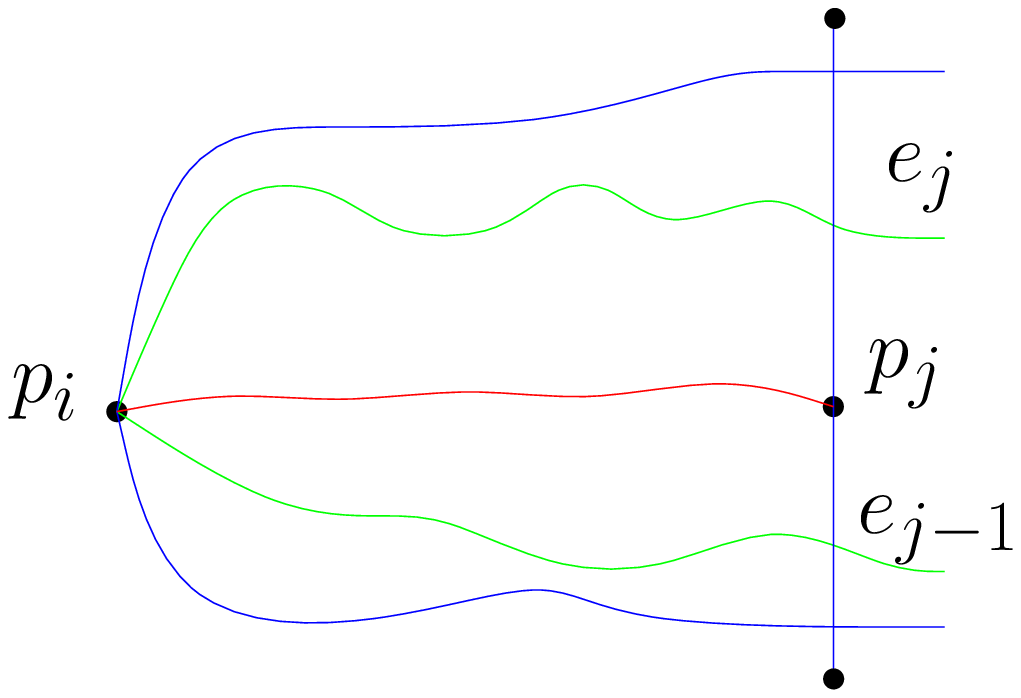} &
\includegraphics[scale=0.3]{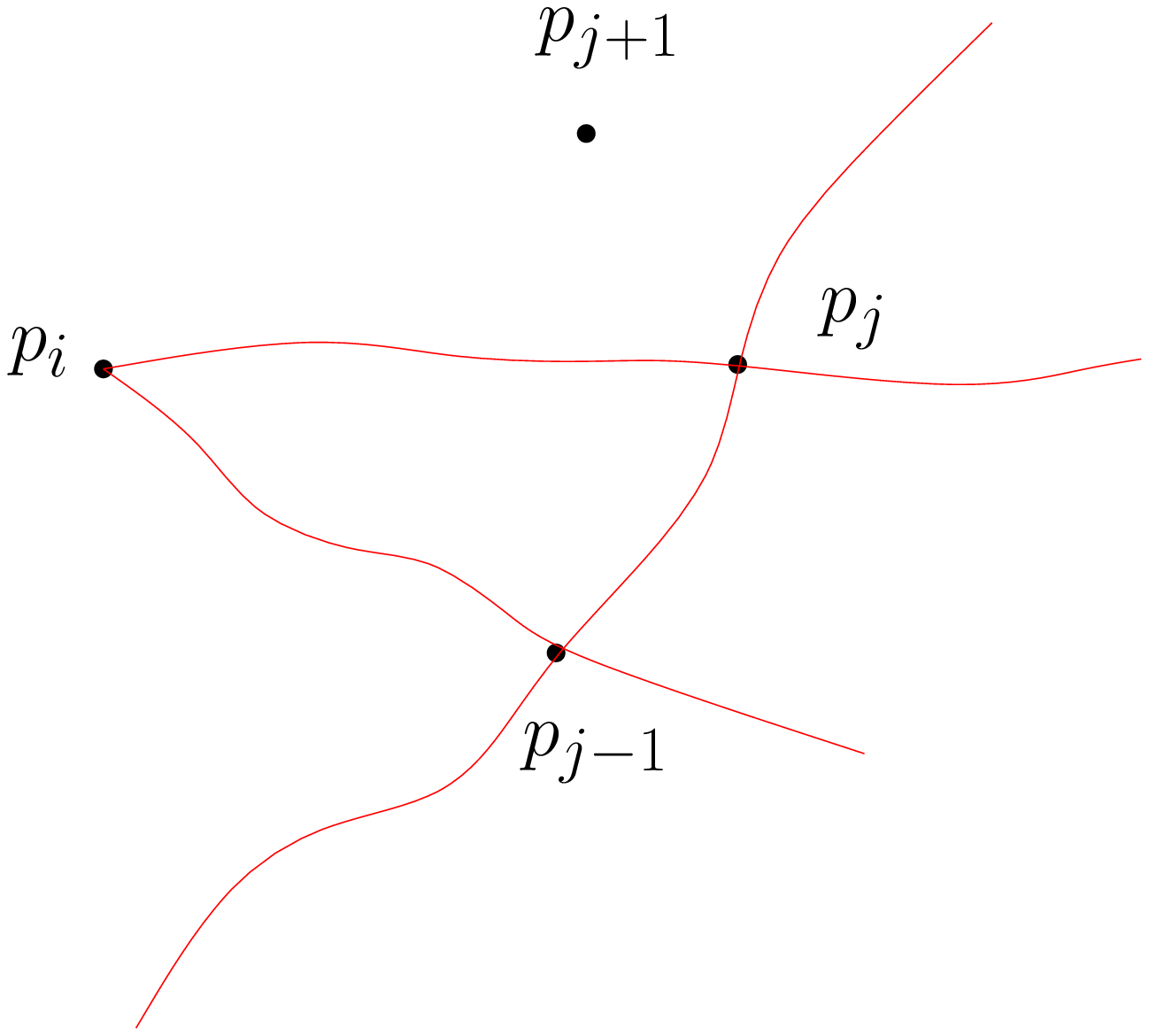} \\
(a)  & (b)  \vspace{1cm} \\ 
\includegraphics[scale=0.4]{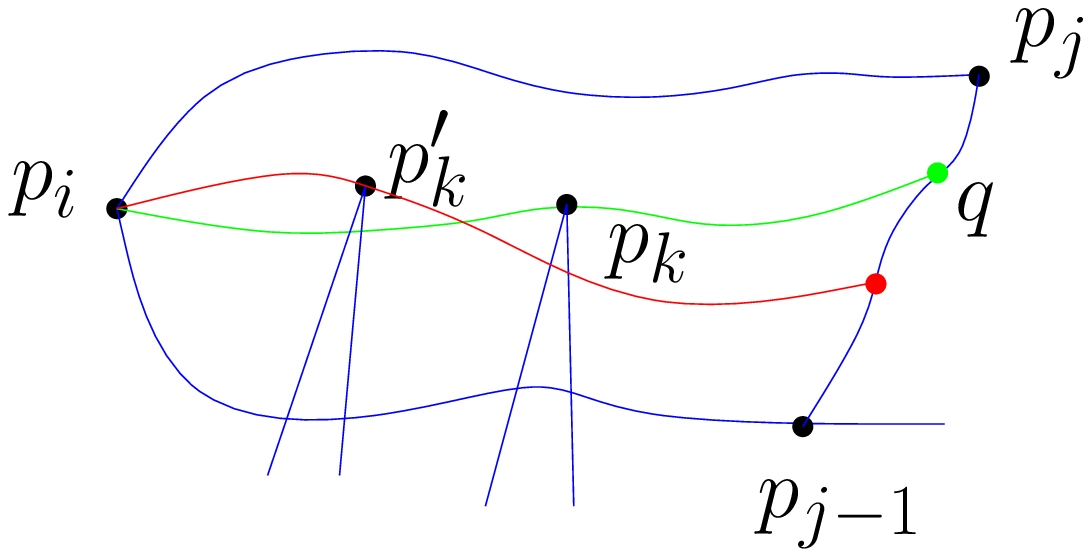} &
\includegraphics[scale=0.3]{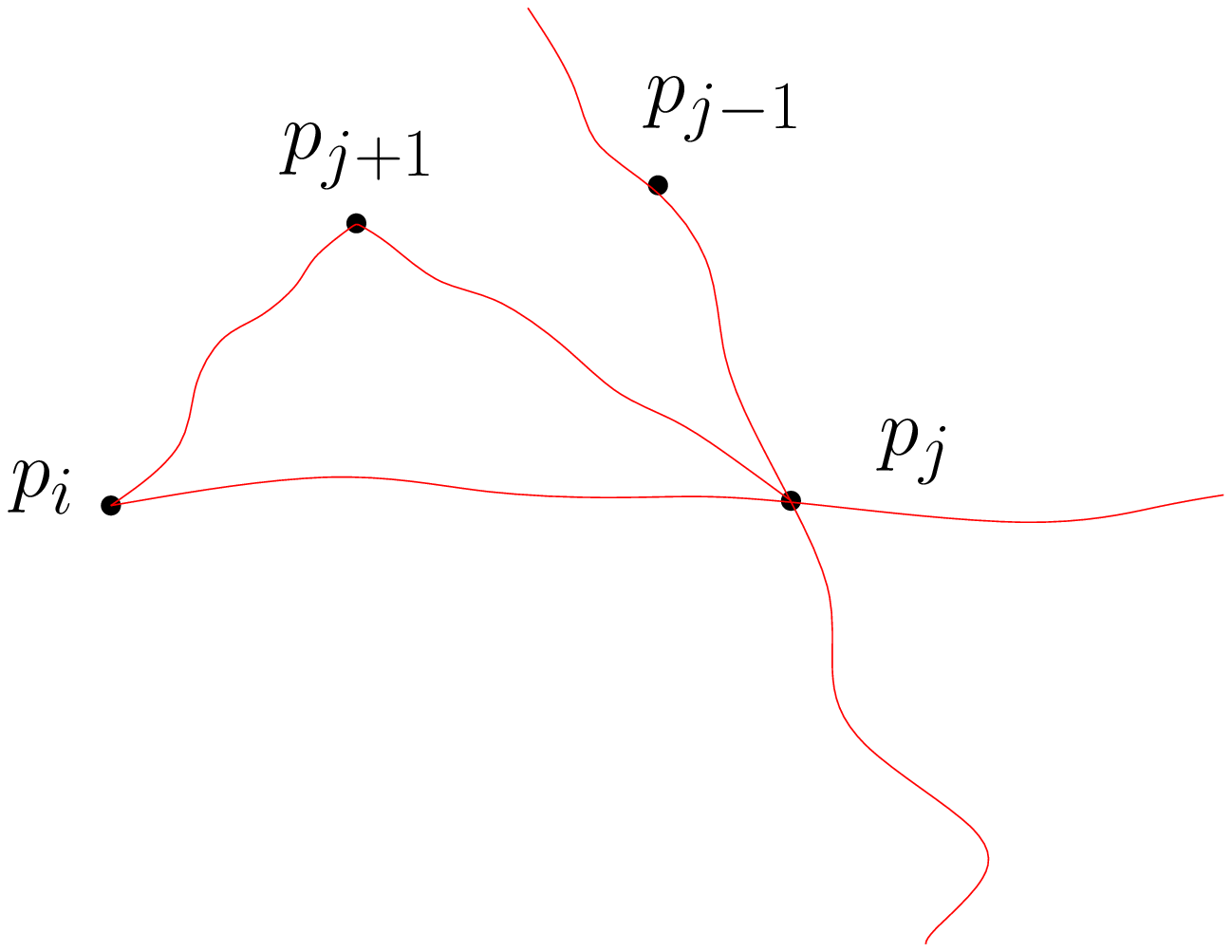} \\
(c)  & (d)  \vspace{1cm} \\ 
\end{tabular}
\caption{An illustration of proof the Lemma \ref{lem:edgeVisibility}.}
\label{fig:lemma1}
\end{figure}

The following lemma from \cite{ORourkeS97} is used in the proof of our key lemma.  Note that Case A and Case B are symmetric.

\begin{lemma}
If $p_k\in \p(p_{b+1}, p_{a-1})$ sees non-adjacent edges $e_a$ and $e_b$ and no edge $\p(p_{a+1}, p_{b})$, then exactly one of Case A or B holds. \textbf{Case A:} (1) $p_k$ sees $p_{a+1}$ but not $p_b$; and (2) $p_{a+1}$ is a witness for ($p_k$, $e_b$); and (3) $p_{a+1}$ sees $e_b$ but $p_b$ does not see $e_a$.  \textbf{Case B:} (1) $p_k$ sees $p_b$ but not $p_{a+1}$; and (2) $p_{b}$ is a witness for ($p_k$, $e_a$); and (3) $p_b$ sees $e_a$ but $p_{a+1}$ does not see $e_b$.
\label{lem:keyEdgeLemma}
\end{lemma}

%Moved this def.  Need to adjust.
We are now ready to present our key structural lemma.

% \begin{figure}
% \centering
% \begin{tabular}{c@{\hspace{0.1\linewidth}}c}
% \includegraphics[scale=0.4]{designatedBlocker.eps} \\
% \end{tabular}
% \caption{An illustration of the designated blocker $p_k$ blocks the invisible pair ($p_i$, $p_j$)}
% \label{fig:designatedBlocker}
% \end{figure}

\begin{lemma}
For any invisible pair $(p_i,p_j)$ in a pseudo-polygon $P$, there is exactly one designated blocking vertex $p_k$.  Moreover, $p_k$ is a candidate blocker for the invisible pair $(p_i,p_j)$ in the visibility graph of $P$.
\label{lem:blockers}
\end{lemma}

\begin{figure}
\centering
\begin{tabular}{c@{\hspace{0.1\linewidth}}c}
\includegraphics[scale=0.4]{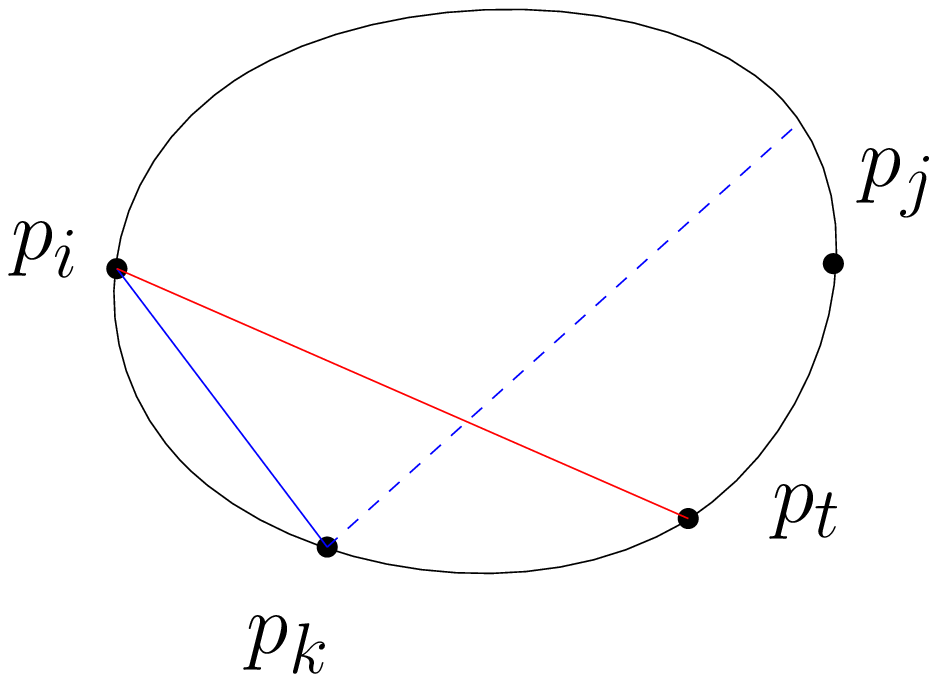} &
\includegraphics[scale=0.4]{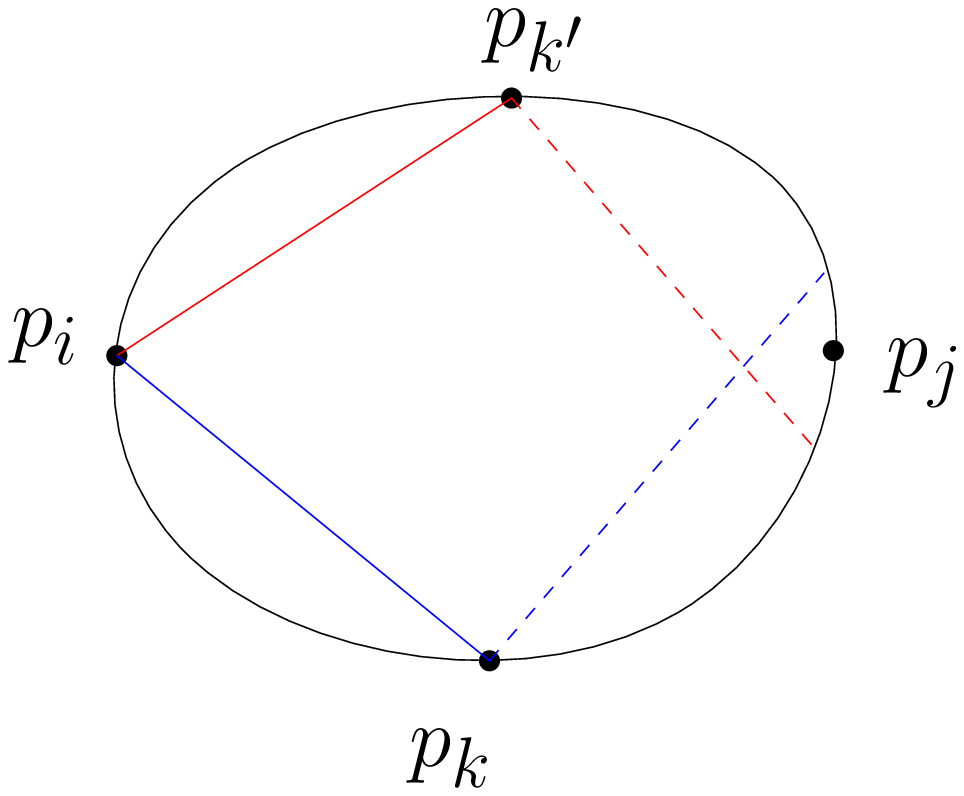} \\
(a)  & (b)  \vspace{1cm} \\ 
\includegraphics[scale=0.4]{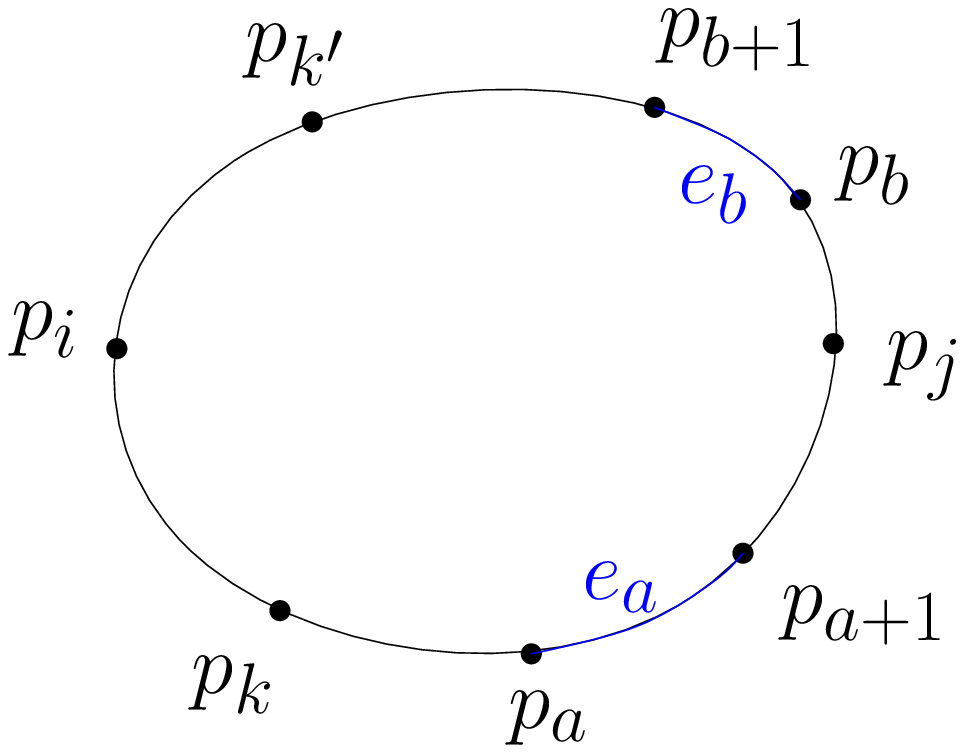} &
\includegraphics[scale=0.4]{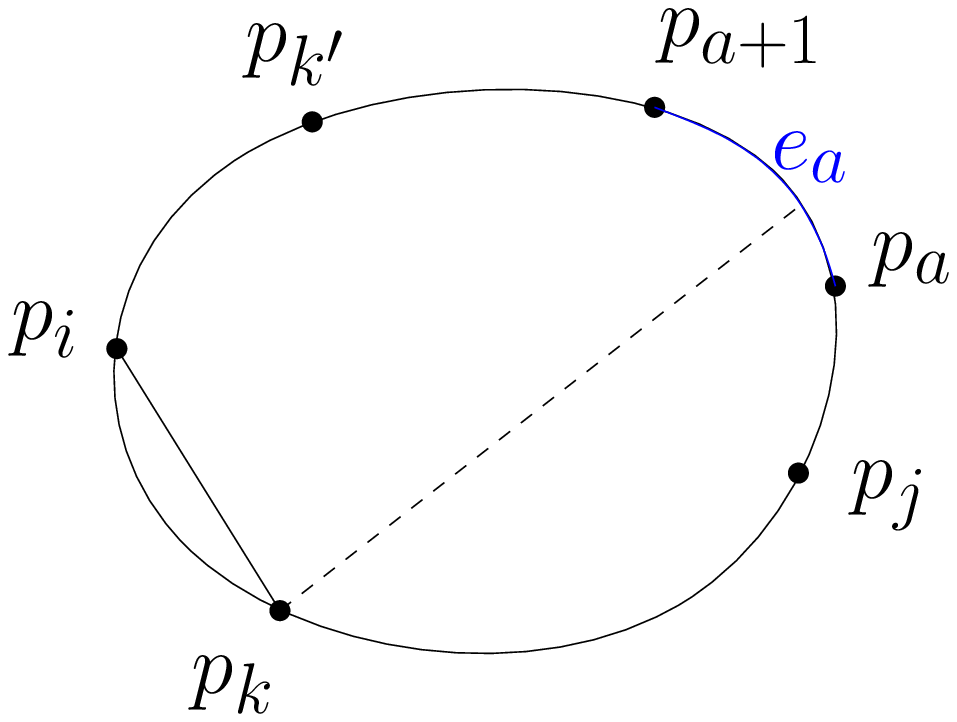} \\
(c)  & (d)  \vspace{1cm} \\ 
\end{tabular}
\caption{An illustration of proof the Lemma \ref{lem:blockers}.}
\label{fig:lemma3}
\end{figure}

\begin{proof}
We begin by showing that a designated blocking vertex $p_k$ for an invisible pair $(p_i,p_j)$ is a candidate blocker for the invisible pair $(p_i,p_j)$.  Without loss of generality, assume that $p_k \in \p(p_i,p_j)$.  For the sake of contradiction, suppose $p_i$ sees a point $p_t \in \p(p_{k+1},p_j)$.  The pseudo-lines $L_{i,k}$ and $L_{i,t}$ intersect at $p_i$, and by the definition of designated blocker, the ray $r_{i,k}$ must intersect $L_{i,t}$ again, a contradiction.  Therefore $p_k$ must be the first point that $p_i$ sees when walking clockwise from $p_j$.  It remains to argue that no point $p_s \in \p(p_{i+1},p_{k-1})$ sees a point $p_t \in \p(p_{k+1},p_j)$.  Suppose the contrary.  Then the segments $p_ip_k$ and $p_sp_t$ must both be contained inside of the polygon, and therefore they must intersect each other, and we also have $r_{i,k}$ must intersect $p_sp_t$ again following the definition of designated blocker, a contradiction. See Figure \ref{fig:lemma3} (a).  It follows that the vertex $p_k$ must be a candidate blocker for the invisible pair $(p_i,p_j)$.

It remains to show that there must be exactly one designated blocker for each invisible pair.  Since each designated blocker is a candidate blocker, there can clearly be at most two designated blockers.  We first show there cannot be two designated blockers for an invisible pair $(p_i,p_j)$.  Suppose $p_k$ and $p_{k'}$ are both designated blockers.  Since they are both candidate blockers, we can assume without loss of generality that $p_k \in \p(p_i,p_j)$ and $p_{k'} \in \p(p_j,p_i)$.  It follows from the definition of designated blocker that $L_{i,k}$ and $L_{i,k'}$ intersect twice. See Figure \ref{fig:lemma3} (b).

We now show that there must be a designated blocker.  Consider an invisible pair $(p_i,p_j)$.  Starting from $p_j$, walk clockwise towards $p_i$ until we reach the first point $p_i$ sees, which we denote $p_k$.  Note that this point must exist since $p_i$ sees $p_{i+1}$.  Similarly walk counter clockwise from $p_j$ until we reach the first point $p_i$ sees, which we denote $p_{k'}$.  Clearly it must be that $p_i$ cannot see any point in $\p(p_{k+1}, p_{k'-1})$.  By Lemma \ref{lem:edgeVisibility} we have that $p_i$ must see at least one edge adjacent to $p_k$ and at least one edge adjacent to $p_{k'}$, and we will show that $p_i$ can see exactly one edge in $\p(p_k,p_{k'})$.  
First suppose that $p_i$ sees no edges in $\p(p_k,p_{k'})$.  Then it must see $e_{k-1}$ and $e_{k'}$ with no edges in $\p(p_k,p_{k'})$.  Applying Lemma \ref{lem:keyEdgeLemma}, we have that either $p_i$ does not see $p_k$ or it does not see $p_{k'}$, a contradiction.  By Lemma \ref{lem:edgeVisibility} we have that $p_i$ cannot see two consecutive edges $e_{s-1}$ and $e_s$ or else $p_i$ would see $p_s \in \p(p_{k+1}, p_{k'-1})$, a contradiction.  So finally suppose $p_i$ sees two non-consecutive edges $e_a$ and $e_b$ in $\p(p_k,p_{k'})$.  Then Lemma \ref{lem:keyEdgeLemma} implies that either $p_i$ sees $p_{a+1}$ or it sees $p_b$, a contradiction in either case.  It follows that $p_i$ must see exactly one edge in $\p(p_k,p_{k'})$. See Figure \ref{fig:lemma3} (c).

Suppose without loss of generality that the edge $e_a \in \p(p_k,p_{k'})$ that $p_i$ sees is in $\p(p_j,p_{k'})$.  Then $p_i$ sees $e_{k-1}$ and $e_a$, and $p_i$ does not see any edge in $\p(p_k,p_{a-1})$.  Applying Lemma \ref{lem:keyEdgeLemma}, we see that we must be in Case A as $p_i$ cannot see $p_a$.  Part (2) from Case A gives us that $p_k$ is a witness for $(p_i,e_a)$, and therefore $r_{i,k}$ first exits the polygon through edge $e_a$.  It follows that $p_k$ is a designated blocker for the invisible pair $(p_i,p_j)$. See Figure \ref{fig:lemma3} (d).  \end{proof}

\section{Necessary Conditions}
\label{sec:nc}
In this section, we give a set of five necessary conditions (NCs) that $G$ must satisfy.  That is, if $G$ does not satisfy one of the conditions then $G$ is \textit{not} the visibility graph for any pseudo-polygon.  Following from Lemma \ref{lem:blockers}, if $G$ is the visibility graph of a pseudo-polygon $P$ then we should be able to assign candidate blockers in $G$ to invisible pairs to serve as the designated blockers in $P$ so that Lemma \ref{lem:blockers} and other pcp properties hold.   The NCs outline a set of properties that this assignment must satisfy if the assignments correspond with a valid set of designated blockers in a pseudo-polygon.  The proofs of these conditions use the definition of designated blockers to show that if the assignment of candidate blockers to invisible pairs do not satisfy the condition, then some pseudo-lines intersect twice, intersect but do not cross, etc.  We illustrate the conditions with simple polygon examples to develop intuition, but the proofs hold for pseudo-polygons.  %The proofs are in the appendix due to lack of space.%The first condition immediately follows from Lemma \ref{lem:blockers}.

%A cycle $w_1, w_2, \ldots, w_k$ in $G$ is an \textit{ordered} cycle if the vertices $w_1, w_2, \ldots, w_k$ follow the order in $C$.  Given any cycle of $G$, an edge that connects two non-adjacent vertices is called a \textit{chord} of the cycle.  The first condition is the same for simple polygon visibility graphs (CITE), but its proof involves geometric arguments that do not apply to pseudo-polygons.  %The fourth condition is a new observation for both pseudo-polygons as well as simple polygons.  

%\begin{nc}
%Every ordered cycle of $k\geq 4$ vertices in a visibility graph $G$ of a pseudo-polygon $P$ has at least $k-3$ chords.  
%\label{nc:chord}
%\end{nc}

Let $(p_i,p_j)$ be an invisible pair, and let $p_k$ be the candidate blocker assigned to it.  The first NC uses the definition of pseudo-lines and designated blockers to provide additional constraints on $p_i$ and $p_k$.  See Fig \ref{fig:OnePairTwoBlocker} (c) for an illustration.  Note that while the condition is stated for $p_k \in \p(p_i,p_j)$, a symmetric condition for when $p_k \in \p(p_j,p_i)$ clearly holds.

%\begin{figure}
%\centering
%\includegraphics[scale=0.5]{newnc1.eps}
%\caption{If $p_k$ is the designated blocker for $(p_i,p_j)$ then it also is for $(p_i,p_t)$.  }
%\label{fig:newnc1}
%\end{figure}

\begin{nc}
%Every invisible pair is assigned exactly one candidate blocker. Additionally, i
If $p_k \in \p(p_i,p_j)$ is the candidate blocker assigned to invisible pair $(p_i,p_j)$ then both of the following must be satisfied: (1) $p_k$ is assigned to the invisible pair $(p_i,p_t)$ for every $p_t \in \p(p_{k+1},p_j)$ and (2) if $(p_k,p_j)$ is an invisible pair then $p_i$ is not the candidate blocker assigned to it.

\label{nc:blockers}
\end{nc}

\begin{proof}
Property (1) easily follows from the definition of designated blockers. See Figure \ref{fig:nc1} (a). Property (2) follows by observing that if this is the case then the pseudo-line $L_{i,k}$ would self-intersect, a contradiction. See Figure \ref{fig:nc1} (b)
 \end{proof}

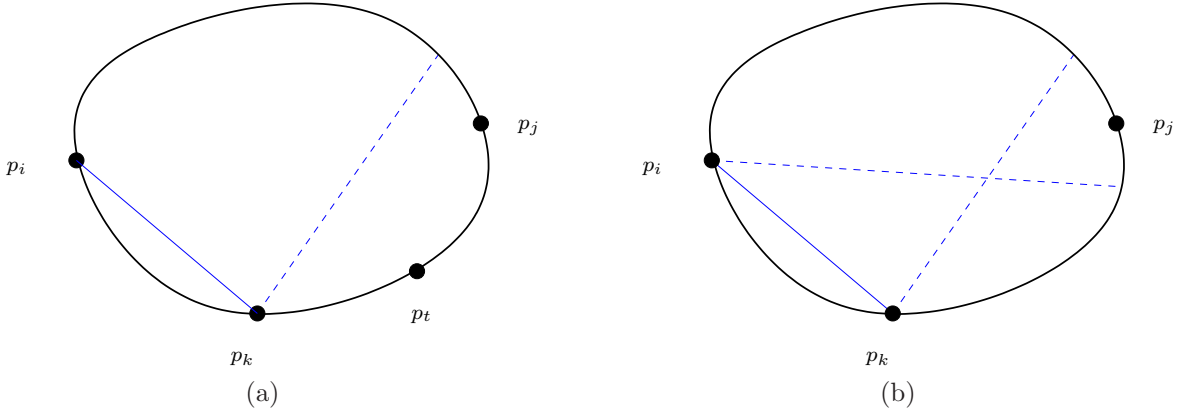
\begin{figure}
\centering
\begin{tabular}{c@{\hspace{0.1\linewidth}}c}
\input{nc1a.pstex_t} &
\input{nc1b.pstex_t}\\
(a) & (b)
\end{tabular}
\caption{Illustrations for Necessary Condition \ref{nc:blockers}.}
\label{fig:nc1}
\end{figure}

%%%%%%%%%%%%%%%%%
Again let $p_k$ be the candidate blocker assigned to an invisible pair $(p_i,p_j)$ such that $p_k \in \p(p_i,p_j)$.  Since $p_k$ is a candidate blocker, we have that $(p_s,p_j)$ is an invisible pair for every $p_s \in \p(p_i, p_{k-1})$.  The next NC is a constraint on the location of designated blockers for $(p_s,p_j)$.  In particular, if $\{p_s,p_k\}$ is a visible pair, then $p_k$ must be the designated blocker for $(p_s,p_j)$.  See Fig \ref{fig:newnc2} (a).  If $(p_s,p_k)$ is an invisible pair, then it must be assigned a designated blocker $p_t$.  In this case, $p_t$ must also be the designated blocker for $(p_s,p_j)$.  See Fig \ref{fig:newnc2} (b).

\begin{figure}
\centering
\begin{tabular}{c@{\hspace{0.1\linewidth}}c}
\includegraphics[scale=0.5]{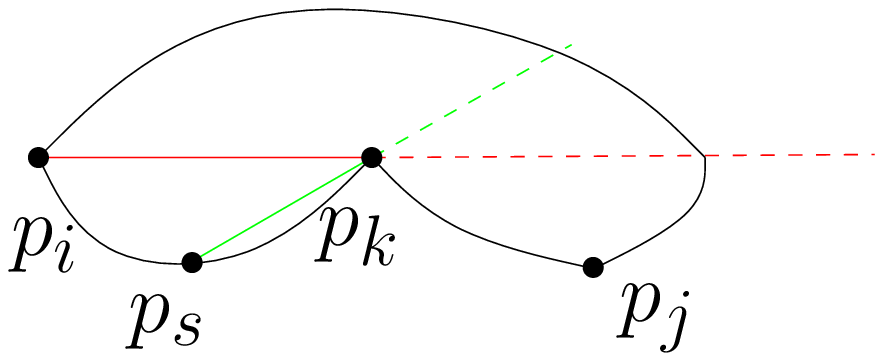}&
\includegraphics[scale=0.5]{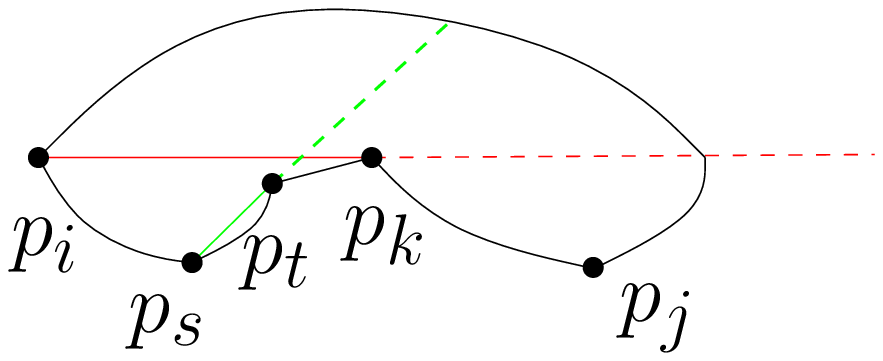}\\
(a) & (b) 
\end{tabular}
\caption{(a) If $p_k$ is the designated blocker for $(p_i,p_j)$ and $p_s$ sees $p_k$ then $p_k$ is the designated blocker for $(p_s,p_j)$.  (b) If $p_s$ does not see $p_k$, and $p_t$ is the designated blocker for $(p_s,p_k)$ then $p_t$ is also the designated blocker for $(p_s,p_j)$. }
\label{fig:newnc2}
\end{figure}

\begin{nc}
\label{nc:alreadyBlocked}
Let $(p_i,p_j)$ denote an invisible pair, and suppose $p_k$ is the candidate blocker assigned to this invisible pair.  Without loss of generality, suppose $p_k \in \p(p_i,p_j)$, and let $p_s$ be any vertex in $\p(p_i,p_{k-1})$.  Then exactly one of the following two cases holds: (1) $\{p_s,p_k\}$ is a visible pair, and the candidate blocker assigned to the invisible pair $(p_s,p_j)$ is $p_k$, or (2) $(p_s,p_k)$ is an invisible pair.  If the candidate blocker assigned to  $(p_s,p_k)$ is $p_t$, then $(p_s,p_j)$ is assigned the candidate blocker $p_t$.
%\begin{enumerate}
%\item $\{p_s,p_k\}$ is a visible pair, and the candidate blocker assigned to the invisible pair $(p_s,p_j)$ is $p_k$.
%\item $(p_s,p_k)$ is an invisible pair.  If the candidate blocker assigned to  $(p_s,p_k)$ is $p_t$, then $(p_s,p_j)$ is assigned the candidate blocker $p_t$.
%\end{enumerate}
\end{nc}

\begin{proof}
Suppose we are in Case 1.  Note that $p_k$ must be the candidate blocker for $(p_s,p_j)$ in $\p(p_s,p_j)$.  Indeed if it were not, then some vertex in $\p(p_s,p_{k-1})$ would have to see a vertex in $\p(p_{k+1},p_j)$ which contradicts that $p_k$ is a candidate blocker for $(p_i,p_j)$.  So if $p_k$ is not assigned to $(p_s,p_j)$ then  the candidate blocker $p_t$ assigned to $(p_s,p_j)$ is in $\p(p_j,p_s)$; however if the corresponding point $p_t$ were the designated blocker, then we would have that the $L_{s,t}$ would intersect the $L_{i,k}$ twice.  This follows because $p_t$ would be the designated blocker for $(p_s,p_j)$ but not $(p_s,p_k)$ (since $p_s$ sees $p_k$), and therefore the $L_{s,t}$ first exits the polygon in $\p(p_k,p_j)$.  Therefore it must be that $p_k$ is assigned to $(p_s,p_j)$.  See Figure \ref{fig:nc2} (a).

Now suppose we are in Case 2, and we have that $(p_s,p_k)$ is an invisible pair which has been assigned candidate blocker $p_t$.  First note that $p_t$ must be in $\p(p_i,p_k)$; it cannot be in $\p(p_k,p_j)$ for the same reasons as the previous case, and it cannot be in $\p(p_{j+1},p_{i-1})$ because $\{p_i,p_k\}$ is a visible pair and such a point could not be a candidate blocker.  If $p_t \in \p(p_i,p_{s-1})$ then Necessary Condition \ref{nc:blockers} implies that $p_t$ must be assigned to $(p_s,p_j)$.  So now suppose that $p_t \in \p(p_{s+1},p_{k-1})$.  If $p_t$ is assigned to $(p_s,p_k)$ but is not assigned to $(p_s,p_j)$ then $p_t$ would be the designated blocker for $(p_s,p_k)$ but not for $(p_s,p_j)$.  It easily follows that $L_{i,k}$ and $L_{s,t}$ intersect twice. See Figure \ref{fig:nc2} (b).
 \end{proof}

\begin{figure}
\centering
\begin{tabular}{c@{\hspace{0.1\linewidth}}c}
\input{nc2a.pstex_t} &
\input{nc2b.pstex_t}\\
(a) & (b)
\end{tabular}
\caption{Illustrations for Necessary Condition \ref{nc:alreadyBlocked}.}
\label{fig:nc2}
\end{figure}
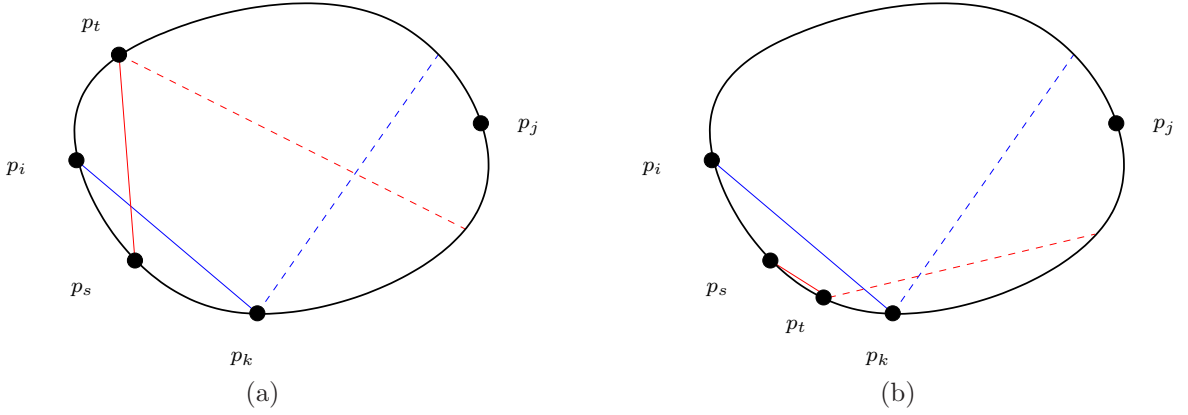

%%%%%%%%%%%%%%%%%%%

The next NC is somewhat similar to Necessary Condition \ref{nc:alreadyBlocked}, except instead of introducing constraints on the designated blockers for $(p_s,p_j)$, it introduces constraints on the designated blockers for $(p_j,p_s)$ (where the order is reversed).  Similar to the previous case, if $p_j$ sees $p_k$ then $p_k$ must block $p_j$ from seeing every $p_s \in \p(p_i,p_{k-1})$, but we can also see that $p_k$ must block $p_j$ from any point $p_t$ such that $p_i$ is the designated blocker for $(p_k,p_t)$.  See Fig \ref{fig:newnc3} (a).  If $p_j$ does not see $p_k$, then there must be a designated blocker $p_q$ for $(p_j,p_k)$.  See Fig \ref{fig:newnc3} (b).  We show that in this case, $p_q$ must be the designated blocker for all $(p_j,p_s)$ and $(p_j,p_t)$.  Also, $(p_i,p_q)$ must be an invisible pair with designated blocker $p_k$.

\begin{figure}
\centering
\begin{tabular}{c@{\hspace{0.1\linewidth}}c}
\includegraphics[scale=0.5]{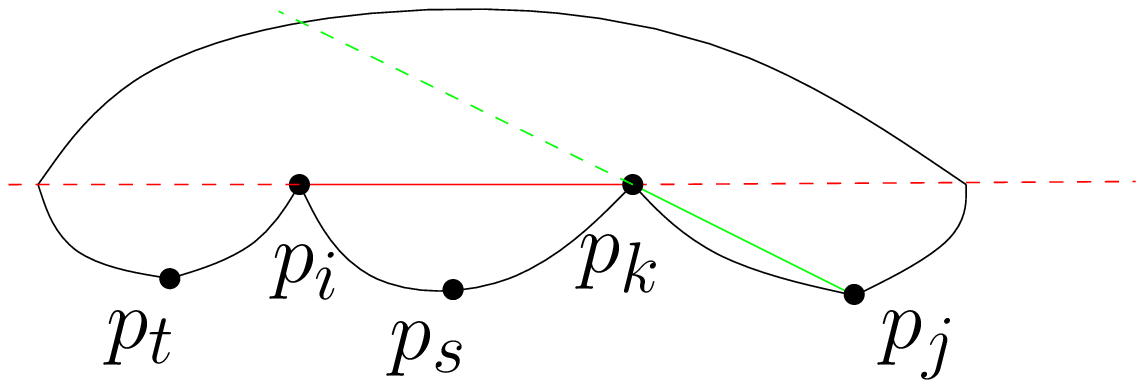}&
\includegraphics[scale=0.5]{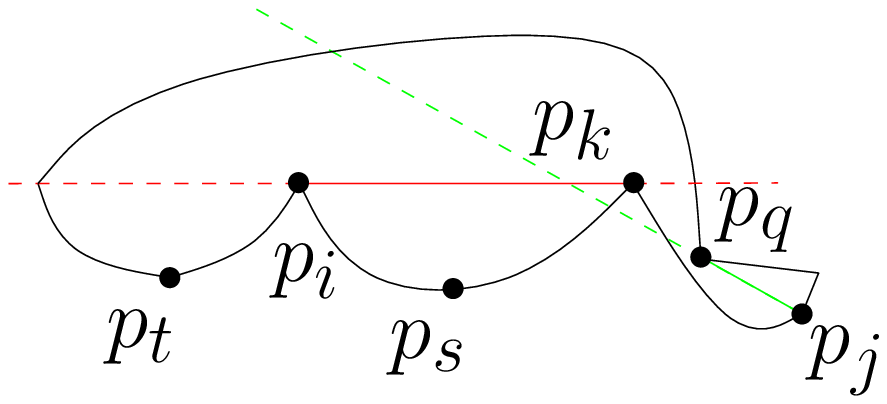}\\
(a) & (b) 
\end{tabular}
\caption{(a) If $p_k$ is the designated blocker for $(p_i,p_j)$ and $p_j$ sees $p_k$ then $p_k$ is the designated blocker for $(p_j,p_s), (p_j, p_i)$, and $(p_j,p_t)$.  (b) If $p_j$ does not see $p_k$, and $p_q$ is the designated blocker for $(p_j,p_k)$ then $p_q$ is the designated blocker for $(p_j,p_s), (p_j, p_i)$, and $(p_j,p_t)$.  Moreover, $(p_i,p_q)$ is an invisible pair and $p_k$ is its designated blocker. }
\label{fig:newnc3}
\end{figure}

\begin{nc}
\label{nc:blockedProperty}
Let $(p_i,p_j)$ denote an invisible pair, and suppose $p_k$ is the candidate blocker assigned to this invisible pair. Without loss of generality, suppose $p_k \in \p(p_i,p_j)$.  Then exactly one of the following two cases holds:

\begin{enumerate}
\item (a) $\{p_j,p_k\}$ is a visible pair. (b) For all $p_s \in \p(p_i,p_{k-1})$, the candidate blocker assigned to the invisible pair $(p_j,p_s)$ is $p_k$. (c) If $p_t$ is such that $p_i$ is the candidate blocker assigned to the invisible pair $(p_k,p_t)$, then $(p_j,p_t)$ is an invisible pair and is assigned the candidate blocker $p_k$.
%\begin{enumerate}
%\item $\{p_j,p_k\}$ is a visible pair.
%\item For all $p_s \in \p(p_i,p_{k-1})$, the candidate blocker assigned to the invisible pair $(p_j,p_s)$ is $p_k$.
%\item If $p_t$ such that $p_i$ is the candidate blocker assigned to the invisible pair $(p_k,p_t)$, then $(p_j,p_t)$ is an invisible pair and is assigned the candidate blocker $p_k$.
%\end{enumerate}

\item (a) $(p_j,p_k)$ is an invisible pair.  Let $p_q$ denote the candidate blocker assigned to $(p_j,p_k)$. (b) $(p_i,p_q)$ is an invisible pair, and $p_k$ is the candidate blocker assigned to it. (c) For all $p_s \in \p(p_i,p_{k})$, the candidate blocker assigned to the invisible pair $(p_j,p_s)$ is $p_q$. (d) If $p_t$ is such that $p_i$ is the candidate blocker assigned to the invisible pair $(p_k,p_t)$, then $(p_j,p_t)$ is an invisible pair and is assigned the candidate blocker $p_q$.

%\begin{enumerate}
%\item $(p_j,p_k)$ is an invisible pair.  Let $p_q$ denote the candidate blocker assigned to $(p_j,p_k)$.
%\item $(p_i,p_q)$ is an invisible pair, and $p_k$ is the candidate blocker assigned to it.
%\item For all $p_s \in \p(p_i,p_{k})$, the candidate blocker assigned to the invisible pair $(p_j,p_s)$ is $p_q$.
%\item If $p_t$ such that $p_i$ is the candidate blocker assigned to the invisible pair $(p_k,p_t)$, then $(p_j,p_t)$ is an invisible pair and is assigned the candidate blocker $p_q$.
%\end{enumerate}

\end{enumerate}
\end{nc}

\begin{proof}
First suppose we are in Case 1: $\{p_j,p_k\}$ is a visible pair.  First note that $(p_j,p_s)$ is an invisible pair for all $p_s \in \p(p_i,p_{k-1})$ or $p_k$ would not be a candidate blocker for $(p_i,p_j)$.   This further implies that $p_k$ is a candidate blocker for $(p_j,p_s)$.  If we assign a candidate blocker $p_a$ in $\p(p_{j+1},p_{s-1})$ to $(p_j,p_s)$ then we have that $L_{i,k}$ and $L_{j,a}$ will intersect twice, as $r_{j,a}$ would first exit the polygon in $\p(p_s,p_k)$.  Therefore we must assign $p_k$ to $(p_j,p_s)$.  See Figure \ref{fig:nc3a} (a).

Now consider a point $p_t$ such that $p_i$ is the candidate blocker assigned to the invisible pair $(p_k,p_t)$.  Note that from property (1) in Necessary Condition \ref{nc:blockers}, we have that $p_t \in \p(p_{j+1},p_k)$, but we just handled the case for all points in $\p(p_i,p_{k-1})$ so we assume that $p_t \in \p(p_{j+1},p_{i-1}$.  From property (2) of Necessary Condition \ref{nc:blockers} we have that  $p_i$ is not assigned to $(p_k,p_j)$, and $p_k$ is not assigned to $(p_k,p_t)$, which implies the rays $r_{i,k}$ and $r_{k,i}$ do not intersect.    If $\{p_j,p_t\}$ were a visible pair then $L_{j,t}$ would intersect $L_{i,k}$ twice, and therefore $(p_j,p_t)$ must be an invisible pair.  We have again that $p_k$ is a candidate blocker for $(p_j,p_t)$.  If a candidate blocker $p_a \in \p(p_j,p_t)$ were used instead then $L_{j,a}$ would intersect $L_{i,k}$ twice, as $r_{j,a}$ would have to first exit the polygon in $\p(p_t,p_k)$. See Figure \ref{fig:nc3a} (b).

\begin{figure}
\centering
\begin{tabular}{c@{\hspace{0.1\linewidth}}c}
\input{nc3a.pstex_t} &
\input{nc3b.pstex_t}\\
(a) & (b)
\end{tabular}
\caption{Illustrations for Case 1 of Necessary Condition \ref{nc:blockedProperty}.}
\label{fig:nc3a}
\end{figure}
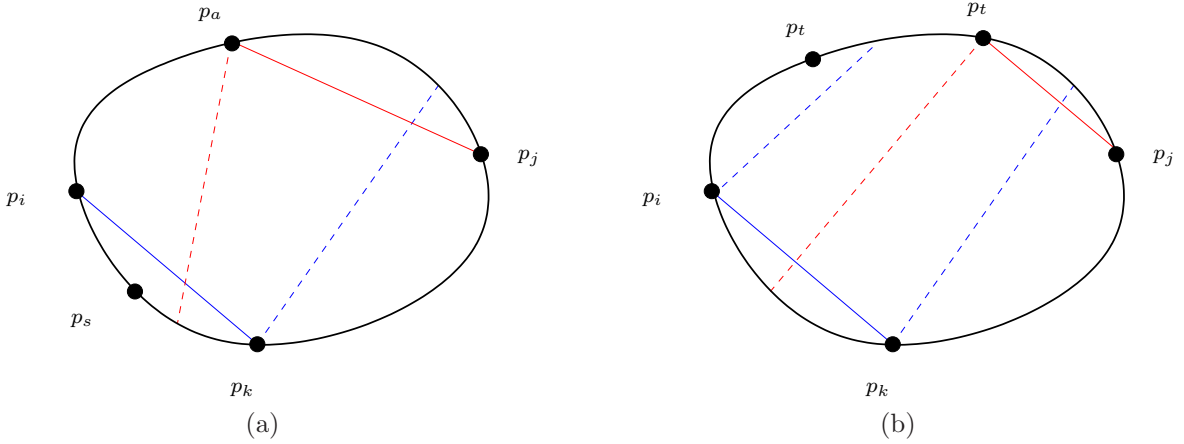

Now suppose we are in Case 2, and we have that $p_q$ is the candidate blocker assigned to the invisible pair $(p_i,p_k)$.  Suppose for the sake of contradiction that $\{p_i,p_q\}$ is a visible pair.  It must be that $p_q \in \p(p_{j+1},p_{i-1})$ since $p_k$ is a candidate blocker for $(p_i,p_j)$, and moreover the ray $r_{i,k}$ first exits $P$ in $\p(p_j,p_q)$. Since $p_q$ is the designated blocker for $(p_j,p_k)$, it follows that $L_{j,q}$ intersects $L_{i,k}$ twice, a contradiction.    So we have that $(p_i,p_q)$ is an invisible pair, and for the sake of contradiction assume that $p_k$ is not the candidate blocker assigned to it.  Then by Necessary Condition \ref{nc:blockers} we have that $p_q \in \p(p_{j+1},p_{i-1})$.  It follows similarly that $L_{j,q}$ would then intersect $L_{i,k}$ twice.  So we have that $(p_i,p_q)$ is an invisible pair, and $p_k$ is the candidate blocker assigned to it. See Figure \ref{fig:nc3b} (a).

Now consider any $p_s \in \p(p_i,p_{k})$.  If $p_q \in \p(p_{j+1},p_{s-1})$, then $p_q$ must be assigned to $(p_j,p_s)$ by Necessary Condition \ref{nc:blockers}.  So now suppose that $p_q \in \p(p_{k+1},p_{j-1})$.  We have that $p_q$ is a candidate blocker for the invisible pair $(p_j,p_s)$, otherwise $p_k$ would not be one for $(p_i,p_j)$.  If we do not assign $p_q$ to $(p_j,p_s)$ then $p_q$ would be the designated blocker for $p_k$ but not for $p_s$ which will cause $L_{j,q}$ and $L_{i,k}$ to intersect twice.  It follows that $p_q$ must be assigned to $(p_j,p_s)$. See Figure \ref{fig:nc3b} (b). %Similarly to the previous case, if we assign a blocker $p_a \in \p(p_{j+1},p_{s-1})$ to $(p_j,p_s)$ then $L_{j,a}$ will intersect $L_{i,k}$ twice.  

%The last part is almost the same as the last part for Case 1.  I'm tired so I'm not writing it now.  Qing said it was true anyway.
So now consider a vertex $p_t$ such that $p_i$ is the candidate blocker assigned to the invisible pair $(p_k,p_t)$.  Similarly as in Case 1, we assume that $p_t \in \p(p_{j+1},p_{i-1})$ and that the rays $r_{i,k}$ and $r_{k,i}$ do not intersect.  Clearly it cannot be that $\{p_j,p_t\}$ is a visible pair or $L_{j,t}$ will intersect $L_{i,k}$ twice.  If $p_q \in \p(p_{j+1},p_{t-1})$ then $p_q$ must be assigned to $(p_j,p_t)$ by Necessary Condition \ref{nc:blockers} since it is assigned to $(p_j,p_k)$.  So suppose $p_q \in \p(p_{k+1},p_j)$.   If $p_q$ is assigned to $(p_j,p_k)$ but is not assigned to $(p_j,p_t)$ then $p_q$ would be the designated blocker for $(p_j,p_k)$ but not for $(p_j,p_t)$.  It easily follows that $L_{i,k}$ and $L_{j,q}$ intersect twice. See Figure \ref{fig:nc3b} (c).
 \end{proof}

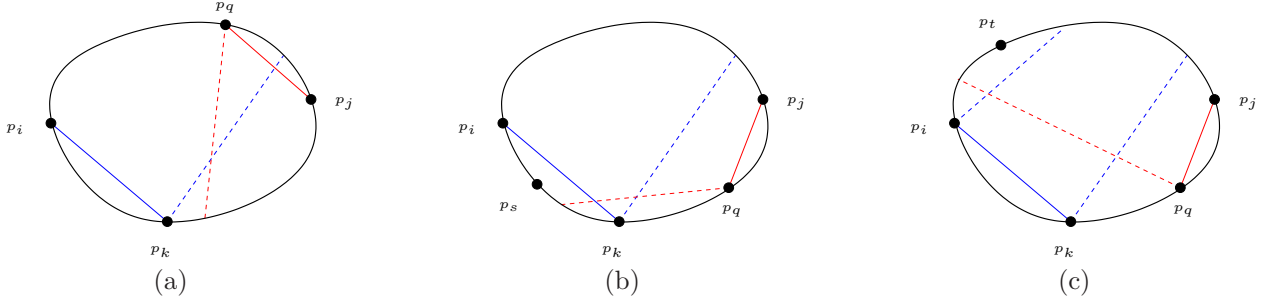
\begin{figure}
\centering
\begin{tabular}{c@{\hspace{0.1\linewidth}}c@{\hspace{0.1\linewidth}}c}
\input{nc3c.pstex_t} &
\input{nc3d.pstex_t}&
\input{nc3e.pstex_t}\\
(a) & (b) & (c)
\end{tabular}
\caption{Illustrations for Case 2 of Necessary Condition \ref{nc:blockedProperty}.}
\label{fig:nc3b}
\end{figure}

%%%%%%%%%%%%%%%%%%%%%%%%

%Let $f(G)$ be an assignment of candidate blockers of $G$ to invisible pairs of $G$.  The remaining three necessary conditions outline properties that $f$ must satisfy.  The next condition is a strengthening of Necessary Condition 2 for simple polygons.

%\begin{nc}
%The assignment $f$ satisfies the following properties: (1) every invisible pair has been assigned exactly one candidate blocker, and (2) if the candidate blocker $p_k$ assigned to $(p_i,p_j)$ is in $\p(p_i,p_j)$ (resp. in $\p(p_j,p_i)$), then the candidate blocker assigned to any invisible pair $(p_{i'},p_{j'})$ such that $p_{i'},p_{j'} \in \p(p_i,p_j)$ (resp.  in $\p(p_j,p_i)$) must be in $\p(p_i,p_j)$ (resp. $\p(p_j,p_i)$).
%\label{nc:blockers}
%\end{nc}

%\begin{obs}
%Let $f(g)$ be any assignment of candidate blockers to invisible pairs that satisfies Necessary Condition \ref{nc:blockers}. If $p_k$ is assigned to the invisible pair $(p_i,p_j)$, and without loss of generality assume $p_k \in \p(p_i,p_j)$.  Then for every $p_l \in \p(p_{k+1},p_j)$, $(p_i,p_l)$ is an invisible pair and $p_k$ is assigned to it.
%\end{obs}

Suppose $p_k$ is a candidate blocker for an invisible pair $(p_i, p_j)$ (or $(p_j, p_i)$), and suppose without loss of generality that $p_i \in \p(p_j,p_k)$.  If $p_k$ is also a candidate blocker for an invisible pair $(p_s, p_t)$ such that $p_s,p_t \in \p(p_k,p_j)$ then we say that the two invisible pairs are a \textit{separable invisible pair}.  We have the following condition which is the same as Necessary Condition 3 for simple polygons in \cite{GhoshG13}.  See Fig \ref{fig:newnc45} (a).

% \begin{figure}
% \centering
% \begin{tabular}{c@{\hspace{0.1\linewidth}}c}
% \includegraphics[scale=0.4]{separableInvisiblePair.eps} \\
% \end{tabular}
% \caption{An illustration of separable invisible pair}
% \label{fig:separableInvisiblePair}
% \end{figure}

%Let $p_i, p_j, p_s,$ and $p_t$ be four distinct vertices, and suppose that $p_k$ is a candidate blocker for two invisible pairs $(p_i,p_j)$ and $(p_s, p_t)$ such that  when ``walking around'' $C$ one encounters $p_i$ and $p_j$ (in either order) before encountering $p_s$ and $p_t$ (in either order).  Then we say that $(p_i,p_j)$ and $(p_s, p_t)$ is a \textit{separable invisible pair}.  Additionally, if $p_k$ is a candidate blocker for (1) $(p_i, p_j)$ or $(p_j, p_i)$ and (2) (p_j, p_s) or $(p_s,p_j) such that 

%Define separable invisible pairs and show that in a p-polygon a vertex can be the designated blocker for at most one such invisible pair.  

\begin{nc}
Suppose $(p_i,p_j)$ and $(p_s,p_t)$ are a separable invisible pair with respect to a candidate blocker $p_k$.  If $p_k$ is assigned to $(p_i,p_j)$ then it is not assigned to $(p_s,p_t)$.
\label{nc:separable}
\end{nc}

\begin{proof}
 If the point $p_k$ that corresponds with $p_k$ is the designated blocker for $(p_i,p_j)$ and $(p_s,p_t)$ then the pseudo-lines $L_{i,k}$ and $L_{s,k}$ intersect at point $p_k$ but do not cross, a contradiction.  See Figure \ref{fig:nc4}.
 \end{proof}

\begin{figure}
\centering
\begin{tabular}{c@{\hspace{0.1\linewidth}}c}
\includegraphics[scale=0.4]{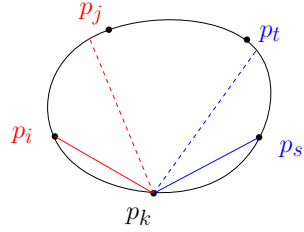} \\
\end{tabular}
\caption{An illustration of Necessary Condition \ref{nc:separable}}
\label{fig:nc4}
\end{figure}

\begin{figure}
\centering
\begin{tabular}{c@{\hspace{0.1\linewidth}}c}
\includegraphics[scale=0.3]{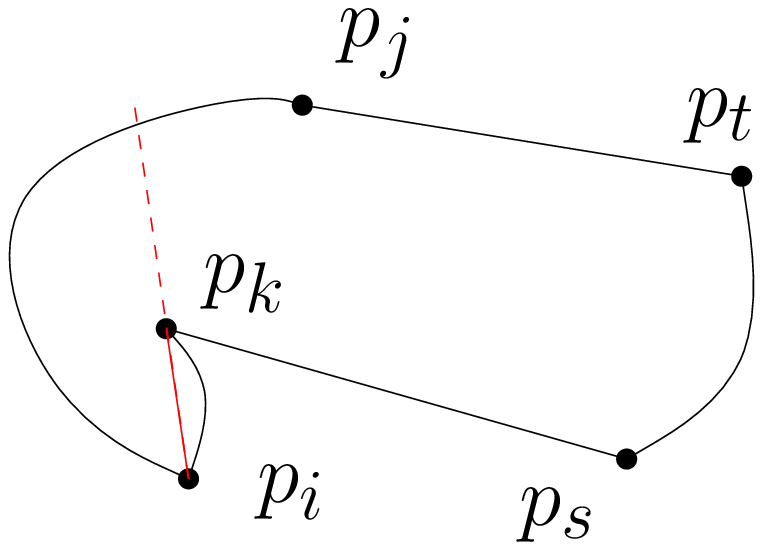}&
\includegraphics[scale=0.3]{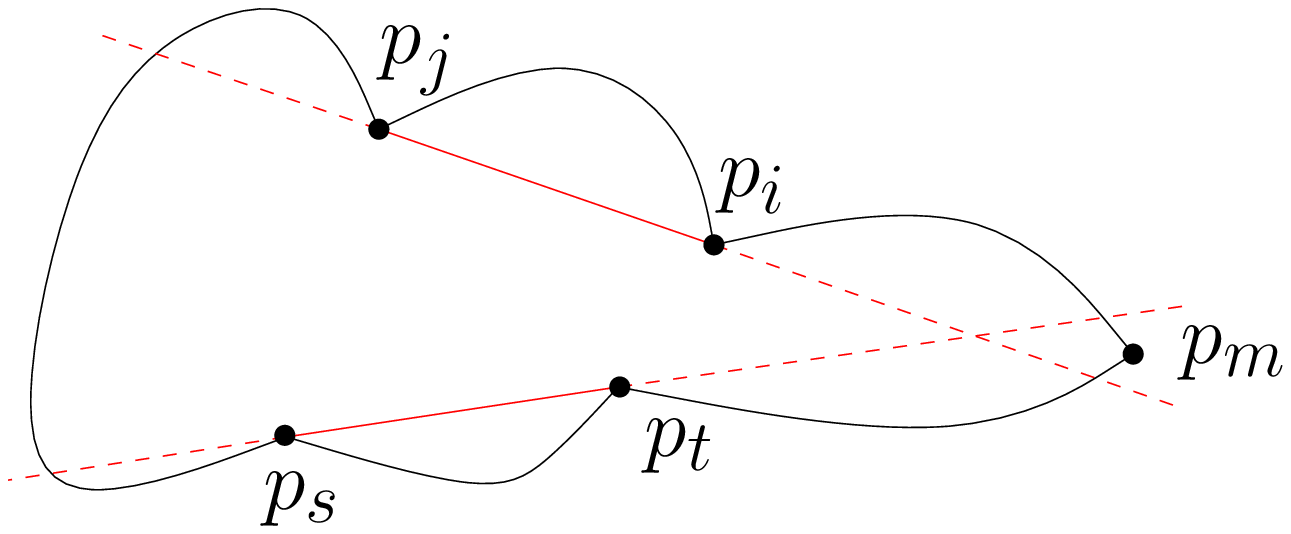}\\
(a) & (b) 
\end{tabular}
\caption{(a) If $p_k$ blocks one invisible pair of a separable invisible pair then it cannot block the other one as well.  (b) $p_i, p_j, p_s$, and $p_t$ are $\{p_i,p_t\}$-pinched.  If $p_j$ blocks $p_i$ from seeing some point, then $p_s$ cannot also block $p_t$ from seeing that point.}
\label{fig:newnc45}
\end{figure}

%The final necessary condition is a generalization of Necessary Condition 4 for simple polygons from (CITE).  

We now give the final NC.  Let $p_i, p_j, p_s,$ and $p_t$ be four vertices of $G$ in ``counter-clockwise order'' around the Hamiltonian cycle $C$.  We say that $p_i, p_j, p_s,$ and $p_t$ are \textit{$\{p_i,p_t\}$-pinched} if there is a $p_m \in \p(p_t,p_i)$ such that $p_i$ is the designated blocker for the invisible pair $(p_j,p_l)$ and $p_t$ is the designated blocker for the invisible pair $(p_s,p_l)$.  See Fig \ref{fig:newnc45} (b).  The notion of $\{p_j,p_s\}$-pinched is defined symmetrically.  %either of the following are true.  

%\begin{itemize}
%\item $(p_j, p_t)$ is an invisible pair and $f$ assigns $p_i$ to it, and $f$ does not assign $p_i$ to $(p_j,p_s)$ (possibly because $\{p_j,p_s\}$ is a visible pair).
%\item there is a $p_l \in \p(p_t,p_i)$ such that $f$ assigns $p_i$ to the invisible pair $(p_j,p_l)$ and $f$ assigns $p_t$ to the invisible pair $(p_s,p_l)$.  
%\end{itemize}

\begin{nc}
Let $p_i, p_j, p_s,$ and $p_t$ be four vertices of $G$ in counter-clockwise order around the Hamiltonian cycle $C$ that are $\{p_i,p_t\}$-pinched.  Then they are not $\{p_j,p_s\}$-pinched.
\label{nc:pinched}
\end{nc}

\begin{proof}
 If $p_i, p_j, p_s,$ and $p_t$ are $\{p_i,p_t\}$-pinched and are $\{p_j,p_s\}$-pinched, then it easily follows from the definition of designated blockers that the pseudo-lines $L_{i,j}$ and $L_{s,t}$ will intersect twice. See Figure \ref{fig:nc5}.
 \end{proof}

\begin{figure}
\centering
\begin{tabular}{c@{\hspace{0.1\linewidth}}c}
\includegraphics[scale=0.4]{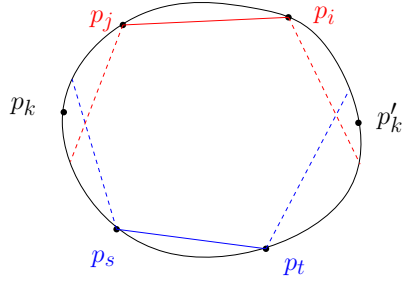} \\
\end{tabular}
\caption{An illustration of Necessary Condition \ref{nc:pinched}}
\label{fig:nc5}
\end{figure}

%Necessary Condition \ref{nc:pinched} is a combinatorial generalization and strengthening of Necessary Condition 4 for simple polygons (CITE), which states that (FILL IN SIMPLE POLYGON NC).

\section{Proving the Conditions are Sufficient}
Suppose we are given an assignment of candidate blockers to invisible pairs that satisfies all NCs presented in Section \ref{sec:nc}.  In this section, we prove that $G$ is the visibility graph for some pseudo-polygon.  We make use of the characterization of vertex-edge visibility graphs for pseudo-polygons given by O'Rourke and Streinu \cite{ORourkeS97}.  That is, we show that the vertex-edge visibility graph associated with $G$ and the assignment of candidate blockers satisfies the necessary and sufficient conditions given in \cite{ORourkeS97}.  %We do so by constructing a pcp $\pcp = (P, \Li)$ in general position such that $G$ is the visibility graph for the corresponding pseudo-polygon.  Abusing notation, we let $P$ denote the pseudo-polygon that corresponds with $\pcp$ as well as its vertices.

We begin by giving an important lemma that relates vertex-edge visibility with designated blockers in any pseudo-polygon $P$.

\begin{lemma}
A vertex $p_i$ does not see an edge $e_j$ if and only if one of the two following conditions hold: (1) $p_s \in \p(p_{i+1},p_j)$ is the designated blocker for $(p_i,p_t)$ for some $p_t \in \p(p_{j+1},p_{i-1})$, or (2) $p_t \in \p(p_{j+1},p_{i-1})$ is the designated blocker for $(p_i,p_s)$ for some $p_s \in \p(p_{i+1},p_{j-1})$.
\label{lem:edgeVis}
\end{lemma}

\begin{proof}
We first show that if (1) or (2) holds, then $p_i$ does not see $e_j$.  Without loss of generality, assume (1) is true.  If $p_i$ were to see $e_j$, then it would need to have two witnesses.  This implies that there would need to be at least two pseudo-lines through $p_i$ that touch $e_j$ before exiting the polygon; however any such line will clearly intersect $L_{i,t}$ twice, a contradiction.  See Figure \ref{fig:lemma4}.  Therefore $p_i$ cannot see $e_j$.

\begin{figure}
\centering
\begin{tabular}{c@{\hspace{0.1\linewidth}}c}
\includegraphics[scale=0.4]{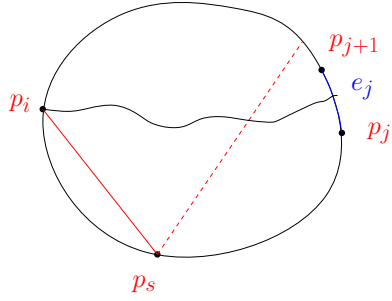} \\
\end{tabular}
\caption{An illustration of Lemma \ref{lem:edgeVis}}
\label{fig:lemma4}
\end{figure}

Now suppose that neither (1) nor (2) holds.  We will show that $p_i$ must see $e_j$.  Then the pseudo-line $L_{i,s}$ for any $p_s \in \p(p_{i+1},p_j)$ must first exit $P$ in $\p(p_i,p_{j+1})$.  If any such pseudo-line first exits $P$ through the interior of $e_j$ (i.e. not through $p_j$), then this $p_s$ will be a witness for $p_i$ and $e_j$.  So suppose that no such pseudo-line first exits $P$ through $e_j$.  This implies that there is no designated blocker for $p_i$ and $p_j$.  It follows from Lemma \ref{lem:blockers} that $p_i$ and $p_j$ must be a visible pair, and therefore $p_j$ is a witness for $p_i$ and $e_j$.  A symmetric argument gives that there is a second witness $p_t \in  \p(p_{j+1},p_{i-1})$ for $p_i$ and $e_j$.  Therefore $p_i$ sees $e_j$.
\end{proof}

Lemma \ref{lem:edgeVis} implies that given any visibility graph $G$ with an assignment of designated blockers to its invisible pairs, there is a unique associated vertex-edge visibility graph.  Let us denote this graph $G_{VE}$.  We will show that if the assignment of designated blockers to the invisible pairs satisfies NCs 1-5, then $G_{VE}$ satisfies the following characterization given by O'Rourke and Streinu \cite{ORourkeS97}.  This implies that there is a pseudo-polygon $P$ such that $G_{VE}$ is the vertex-edge visibility graph of $P$ \textit{and} $G$ is the visibility graph of $P$.  Note $p_j$ is an \textit{articulation point} of the subgraph of $G_{VE}$ induced by $\p(p_{i+1},p_k)$ if and only if $p_j$ is a candidate blocker for the invisible pair $(p_s,p_{t})$ for some $p_s \in \p(p_{i+1},p_{j-1})$ and some $p_t \in \p(p_{j+1},p_k)$. 

\begin{theorem}
\cite{ORourkeS97}
A graph is the vertex-edge visibility graph of a pseudo-polygon $P$ if and only if it satisfies the following.  If $p_k \in \p(p_{j+1},p_{i-1})$ sees two non-adjacent edges $e_i$ and $e_j$ and no edge in $\p(p_{i+1},p_{j-1})$ then it satisfies exactly one of the following two properties: (1) $p_{i+1}$ sees $e_j$ and $p_{i+1}$ is an articulation point of the subgraph induced by $\p(p_k,p_j)$, or (2) $p_j$ sees $e_i$ and $p_j$ is an articulation point of the subgraph induced by $\p(p_{i+1},p_k)$.
\label{thm:veChar}
\end{theorem}

\noindent\textbf{Good Lines and Centers.}  If $L_{i,j}$ is such that $\{p_i,p_j\}$ is a visible pair, then we say $L_{i,j}$ is a \textit{good line}.  Recall $L_{i,j}$ can be decomposed into three portions: the \textit{segment} $p_ip_j$ and two infinite \textit{rays} $r_{i,j}$ and $r_{j,i}$. The ray $r_{i,j}$ starts at $p_j$ and does not include $p_i$, and $r_{j,i}$ is defined symmetrically.  We now define the \textit{center} of $L_{i,j}$ to be the connected subsegment of $L_{i,j}$ consisting of the following: the segment $p_ip_j$, the subsegment of $r_{i,j}$ obtained by starting at $p_j$ and walking along the ray until we first reach exit outside of $P$ (this may or may not be just $p_j$), and the symmetric subsegment of $r_{j,i}$.  Note that the center of $L_{i,j}$ is simply the intersection of $L_{i,j}$ and $P$  if the rays never re-enter $P$ after leaving.%for the first time. 

% \begin{figure}
% \centering
% \begin{tabular}{c@{\hspace{0.1\linewidth}}c}
% \includegraphics[scale=0.4]{constructingGoodLines(1).eps} &
% \includegraphics[scale=0.4]{constructingGoodLines(2).eps} \\
% (a)  & (b)   \vspace{1cm} \\ 
% \includegraphics[scale=0.4]{constructingGoodLines(3).eps} \\
%  (c)  \vspace{1cm} \\ 
% \end{tabular}
% \caption{An illustration of Constructing the Center of Good Lines.}
% \label{fig:constructingGoodLines}
% \end{figure}

Given the visibility graph $G$ and the assignment of candidate blockers to invisible pairs, we will now describe how to construct a \textit{witness} $P'$ that will be used to show that $G$ is the visibility graph of a pseudo-polygon $P$.  $P'$ has a vertex for each vertex of $G$, and for every visible pair $\{p_i,p_j\}$ in $G$, the center of $L_{i,j}$ will appear in $P'$.  The center will behave according to the assignment of candidate blockers to invisible pairs.  In other words, if $p_j$ is assigned to the invisible pair $(p_i,p_k)$, then the center will be defined so that it fits the definition of designated blocker for this invisible pair. 

For each vertex $p_i$ in $G$, we add a point $p_i$ to $P'$.  We place these points in $\mathbb{R}^2$ in convex position in ``counterclockwise order''.  That is, indices increase (modulo $n$) when walking around the convex hull in the counterclockwise direction.   Now suppose that $p_j$ is the candidate blocker assigned to an invisible pair $(p_i,p_s)$.  We define $r_{i,j}$ to be such that $p_j$ is a designated blocker for $(p_i,p_s)$.  First note that if $p_j$ is the candidate blocker assigned to $(p_i,p_s)$ and $(p_i,p_t)$, then it cannot be that one of $p_s$ and $p_t$ is in $\p(p_i,p_j)$ and the other is in $\p(p_j,p_i)$ by Necessary Condition \ref{nc:separable}, so without loss of generality assume that any such point is in $\p(p_j,p_i)$.  Let $p_s$ be such that $p_j$ is assigned to $(p_i,p_s)$ but it is not assigned to $(p_i, p_{s+1})$.  It follows from Necessary Condition \ref{nc:blockers} that there is exactly one such point $p_s$ that satisfies this condition.  We begin the definition of $r_{i,j}$ as a straight line from $p_j$ to the edge $e_s$.  There may be many rays from many different vertices which intersect the edge $e_s$.  If $r_{a,b}$ is another ray intersecting $e_s$, we ``preserve the order'' of the rays so that $r_{i,j}$ and $r_{a,b}$ do not intersect.  Note that because of property (2) of Necessary Condition \ref{nc:blockers}, these centers do not self-intersect.  %We now have the following lemma.

\begin{lemma}
If $G_{VE}$ does not satisfy the conditions of Theorem \ref{thm:veChar}, then there exists a pair of distinct good line centers that intersect twice in $P'$.
\label{lem:violateChar}
\end{lemma}

\begin{proof}
Suppose that in the graph, we have a vertex $p_k\in \p(p_{j+1},p_{i-1})$ sees two non-adjacent edges $e_i$ and $e_j$ and no edge in $\p(p_{i+1},p_{j-1})$ but the graph does not satisfy (1) or (2).  We will first show that since $p_k$ sees two non-adjacent edges $e_i$ and $e_j$,  Lemma \ref{lem:edgeVis} implies that either $p_{i+1}$ is the designated blocker for $(p_k,p_j)$ but not for $(p_k,p_{j+1})$ or $p_j$ is the designated blocker for $(p_k,p_{i+1})$ but not $(p_k,p_i)$ (in either scenario, $\{p_k,p_{j+1}\}$ and/or $\{p_k,p_i\}$ may actually be a visible pair).  In other words, either $p_{i+1}$ or $p_j$ blocks $p_k$ from seeing all edges in $\p(p_{i+1},p_j)$.  Indeed if the designated blocker for any such edge was in $\p(p_k,p_i)$ then $p_k$ would not see $e_i$, and if it were in $\p(p_{j+1},p_k)$ then $p_k$ would not see $e_j$.  Additionally, no vertex in $\p(p_{i+2},p_{j-1})$ can block the edges as $p_k$ cannot see any such vertices by Lemma \ref{lem:edgeVisibility}.  It follows that $p_{i+1}$ and $p_j$ are the only points that can block $p_k$ from seeing these edges.  For any such edge $e_y$, it cannot be that $p_{i+1}$ and $p_j$ both block $p_k$ from $e_y$ as $(p_k,p_y)$ would have two designated blockers, contradicting Necessary Condition 1.  This implies that the ``edge intervals'' blocked by $p_{i+1}$ and $p_j$ cannot overlap, and therefore if they each block some of the edges, then $p_k$ would see an edge in $\p(p_{i+1},p_{j-1})$, a contradiction.  So we can suppose without loss of generality that $p_{i+1}$ blocks $p_k$ from seeing all edges in $\p(p_{i+1},p_{j-1})$ but not edge $e_j$.  It follows from Lemma \ref{lem:edgeVis} that $p_{i+1}$ is the designated blocker for $(p_k,p_j)$ but not $p_{j+1}$, and therefore the center of $L_{k,i+1}$ will intersect $e_j$.

Since $G_{VE}$ does not satisfy Theorem \ref{thm:veChar}, it must be that neither condition (1) nor (2) holds.  Since $p_{i+1}$ is the designated blocker for $(p_k,p_j)$, it follows that $p_{i+1}$ is an articulation point of the subgraph of $G_{VE}$ induced by $\p(p_k,p_j)$.  Then $p_{i+1}$ must not see $e_j$ or else (1) would hold.  If $p_{i+1}$ does not see $e_j$, Lemma \ref{lem:edgeVis} implies that there is a vertex $p_s \in \p(p_{i+2},p_{j})$ that is the designated blocker for $(p_{i+1},p_j)$ or there is a vertex $p_t \in \p(p_{j+1},p_i)$ that is the designated blocker for $(p_{i+1},p_j)$.  In either case, the center of the good line that blocks $p_{i+1}$ from $e_j$ would necessarily cross the center of $L_{k,i+1}$ twice.  
 \end{proof}

Combining Lemma \ref{lem:violateChar} with the following lemma, we get that $G_{VE}$ satisfies Theorem \ref{thm:veChar} and therefore is the vertex-edge visibility graph for a pseudo-polygon.

\begin{lemma}
The centers of any pair of good lines intersects at most once, and if they intersect they cross.
\label{lem:goodCross}
\end{lemma}

\begin{proof}

Let $\{p_i,p_k\}$ and $\{p_s,p_t\}$ denote any two visible pairs.  We will show that the centers of the segments intersect at most once, and if they do then they cross.  We consider three main subcases: (1) the four points are distinct and the segment $p_ip_k$ intersects the segment $p_sp_t$, (2) the four points are distinct and $p_ip_k$ does not intersect $p_sp_t$, and (3) the visible pairs share a point (i.e. $p_k = p_t$).  Throughout the proof, we use the fact that neither center self-intersects. 

\textbf{Case 1:}  First, suppose the segments $p_ip_k$ and $p_sp_t$ intersect each other.  We will prove that the centers do not intersect again. First note that $p_t$ cannot block $p_s$ from seeing $p_i$ or $p_k$, because if $p_t$ is not a candidate blocker for $(p_s,p_i)$ and $(p_s,p_k)$ because \{$p_i$, $p_k$\} is a visible pair.  Therefore $r_{s,t}$ cannot intersect $p_ip_k$.  See Figure \ref{fig:GroupOne} (a). Symmetric arguments show that $r_{t,s}$ cannot intersect $p_ip_k$ and that $r_{i,j}$ and $r_{j,i}$ cannot intersect $p_sp_t$.  It follows that if there is a second intersection, it must be a ray of $L_{i,k}$ intersecting with a ray of $L_{s,t}$.  Then without loss of generality, suppose $p_k$ blocks $p_i$ from seeing a point $p_j$. Since $r_{i,k}$ cannot intersect segment $p_sp_t$, without loss of generality, assume $p_j \in \p(p_{k+1}, p_{t-1})$. Clearly $r_{t,s}$ cannot intersect $r_{i,k}$ without intersecting $p_ip_k$, and suppose that $p_t$ is the designated blocker for $(p_s,p_j)$. By Necessary Condition \ref{nc:alreadyBlocked}, we have that the candidate blocker for ($p_s$, $p_j$) is either $p_k$ or the candidate blocker assigned to ($p_s$, $p_k$), but in this case $p_t$ is not the candidate blocker assigned to ($p_s$, $p_k$), a contradiction. See Figure \ref{fig:GroupOne} (b).

\begin{figure}
\centering
\begin{tabular}{c@{\hspace{0.1\linewidth}}c}
\includegraphics[scale=0.4]{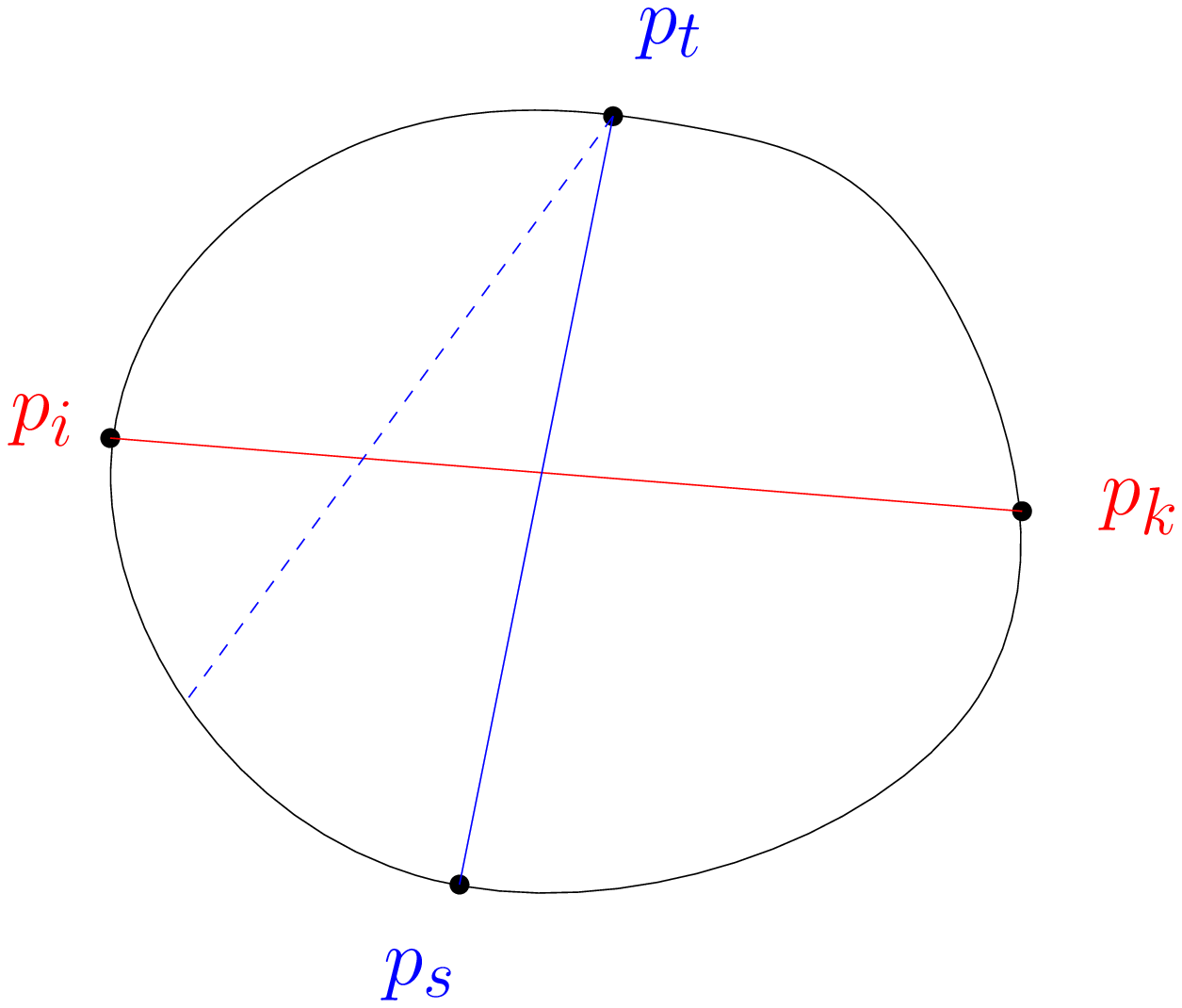} &
\includegraphics[scale=0.4]{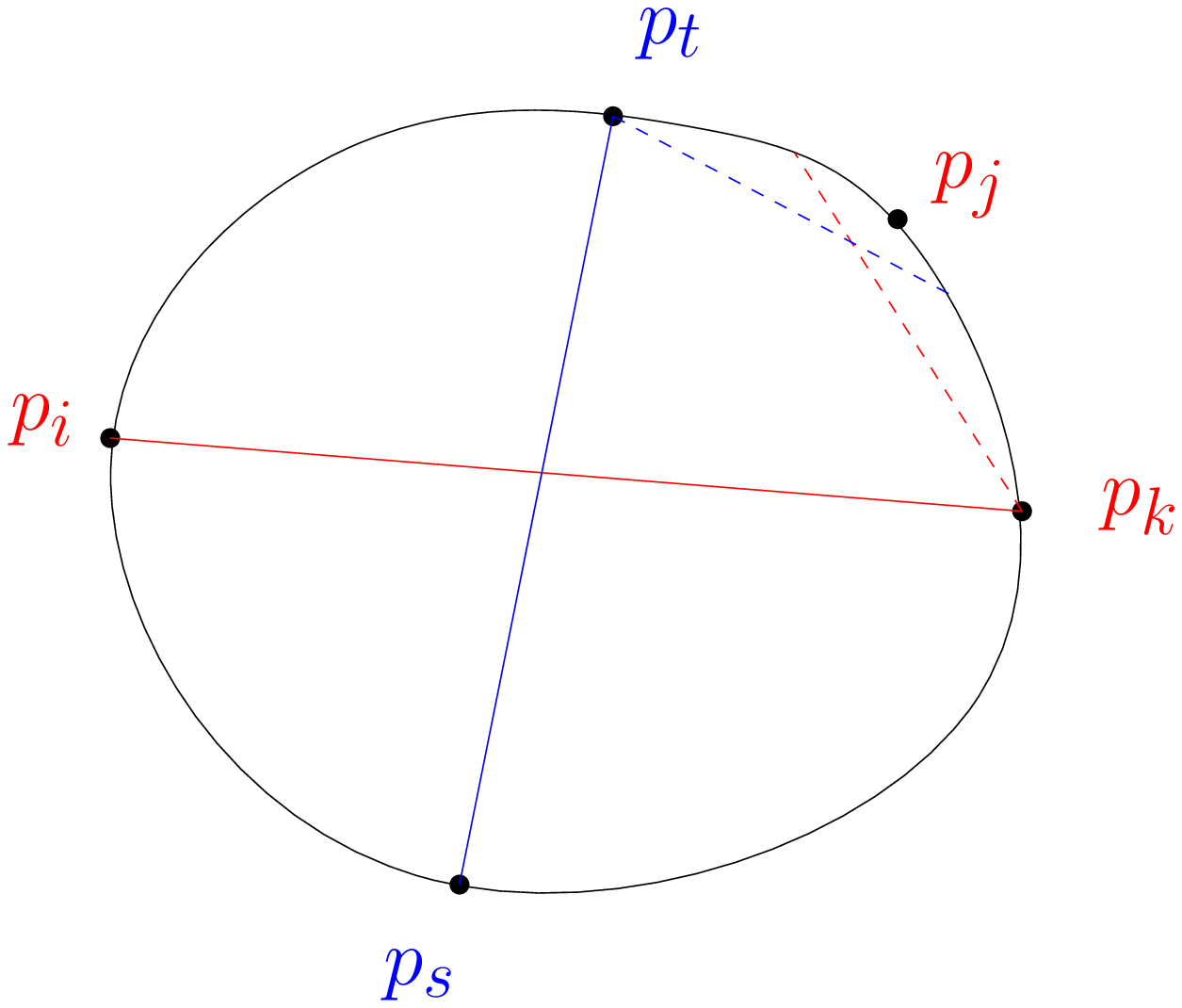} \\
(a)  & (b)   \vspace{1cm} \\ 
\end{tabular}
\caption{An illustration of segments intersect each other.}
\label{fig:GroupOne}
\end{figure}

\textbf{Case 2.}  Suppose the four points are distinct and satisfy $p_s \in \p(p_i,p_t)$ and $p_k \in \p(p_t,p_i)$.  It follows that $p_ip_k$ does not intersect $p_sp_t$, and we will prove that the corresponding centers do not intersect more than once. Suppose that $p_t$ blocks $(p_s,p_k)$, but does not block $(p_s,p_i)$. Note that this implies that $r_{s,t}$ intersects $p_ip_k$, and we will show that they do not intersect again.  We first will show that $p_s$ cannot block $(p_t,p_i)$.  See Figure \ref{fig:GroupTwo} (a).  If it does, then Necessary Condition \ref{nc:blockedProperty} implies that ($p_i$, $p_k$) is an invisible pair, a contradiction.  

We now show that if $p_t$ blocks $(p_s,p_k)$ but not $(p_s,p_i)$, then $p_i$ cannot also block $(p_k,p_s)$. See Figure \ref{fig:GroupTwo} (b).  Suppose for the sake of contradiction that $p_i$ does block $(p_k,p_s)$.  It follows by Necessary Condition \ref{nc:blockedProperty} that $(p_s,p_k)$ must be blocked either by $p_i$ or the point that blocks $(p_s,p_i)$.  But since $p_t$ blocks $(p_s,p_k)$, that implies that $p_t$ must also block $(p_s,p_j)$, a contradiction.  See Figure \ref{fig:GroupTwo} (c). %Then it followe have that the candidate blocker assigned to the invisible pair ($p_k$, $p_s$) is $p_i$, and ($p_k$, $p_t$) is an invisible pair and $p_i$ is the candidate blocker assigned to ($p_k$, $p_t$).  By \ref{nc:blockedProperty} case 2, the candidate blocker assigned to the invisible pair ($p_s$, $p_k$) is $p_t$, and ($p_s$, $p_i$) is an invisible pair and $p_t$ is the candidate blocker assigned to ($p_s$, $p_i$). See Figure \ref{fig:GroupTwo} (c). 

 Next we show that if $p_t$ blocks $(p_s,p_k)$  but not $(p_s,p_i)$, then $p_k$ cannot block $(p_i,p_t)$ but not $(p_i,p_s)$. See Figure \ref{fig:GroupTwo} (d).  If so, then ($p_t$, $p_i$) is an invisible pair.  We will show that $p_k$ cannot be the designated blocker for $(p_t,p_i)$.  If it is, then ($p_s$, $p_i$) is an invisible pair and is assigned the candidate blocker $p_t$ by Necessary Condition \ref{nc:blockedProperty} case 1, a contradiction. So $p_k$ is not the candidate blocker assigned to the invisible pair ($p_i$, $p_t$). By Necessary Condition \ref{nc:blockedProperty}, we have that if the candidate blocker assigned to the invisible pair ($p_i$, $p_t$) is not $p_k$, then ($p_t,p_k$) is an invisible pair. Let $p_q$ denote the point that blocks ($p_t$, $p_k$). By Necessary Condition \ref{nc:blockedProperty} case 2 part (b), we have that ($p_i$, $p_q$) is an invisible pair, and $p_k$ is the candidate blocker assigned to it; therefore $p_q \in \p (p_{s+1}, p_{k-1})$. Now we will show that $p_q$ actually cannot be in $\p (p_{s+1}, p_{t-1})$. If $p_q \in \p (p_{s+1}, p_{t-1})$, then $p_q$ is not a candidate blocker for invisible pair ($p_t$, $p_k$), because $\{p_s,p_t\}$ is a visible pair. So now we have that $p_q \in \p (p_{t+1}, p_{k-1})$ and $p_q$ blocks $p_t$ from seeing $p_i$. By Necessary Condition \ref{nc:blockedProperty}, we have that ($p_s$, $p_q$) is an invisible pair, and $p_t$ is the candidate blocker assigned to it. Then we have that $p_t$ blocks $p_s$ from seeing $p_i$ by Necessary Condition \ref{nc:blockedProperty}, a contradiction. 

Now suppose that $p_t$ blocks $(p_s,p_k)$ but not $(p_s,p_i)$, and $p_k$ blocks $(p_i,p_j)$ for some $p_j \in \p(p_{k+1}, p_{i-1})$.  We will show that $p_t$ is the candidate blocker assigned to the invisible pair ($p_s$, $p_j$) and therefore $r_{s,t}$ does not intersect $r_{i,k}$. By Necessary Condition \ref{nc:alreadyBlocked} case 2, ($p_s$, $p_k$) is an invisible pair. Since the candidate blocker assigned to ($p_s$, $p_k$) is $p_t$, then ($p_s$, $p_j$) is assigned the candidate blocker $p_t$. See Figure \ref{fig:GroupTwo} (e). 

Suppose $p_t$ blocks $(p_s,p_k)$ but not $(p_s,p_i)$, and $p_i$ blocks $(p_k,p_j)$ for some point $p_j \in \p(p_{k+1}, p_{i-1})$.  We will show that $p_t$ cannot block $p_s$ from seeing $p_j$. See Figure \ref{fig:GroupTwo} (f). By Necessary Condition \ref{nc:alreadyBlocked}, we have that the designated blocker for $(p_s,p_j)$ is $p_i$ if $\{p_s,p_i\}$ is a visible pair and otherwise is the designated blocker for $(p_s,p_i)$.  But we assumed $p_t$ is not the designated blocker for $(p_s,p_i)$ and therefore cannot block $(p_s,p_j)$.%, and the candidate blocker assigned to the invisible pair ($p_s$, $p_j$) is $p_k$, if ($p_s$, $p_i$) is an invisible pair, then if the candidate blocker assigned to ($p_s$, $p_i$) is $p_t$, then ($p_s$, $p_j$) is assigned the candidate blocker $p_t$. A contradiction. 

Now suppose $p_t$ blocks $(p_s,p_k)$ but not $(p_s,p_i)$, and $p_s$ blocks $p_t$ from seeing a point $p_j \in \p(p_{i+1}, p_{s-1})$.  We will show that $p_i$ cannot block $(p_k,p_j)$. See Figure \ref{fig:GroupTwo} (g).  By Necessary Condition \ref{nc:blockedProperty}, ($p_k$, $p_j$) is an invisible pair, the candidate blocker assigned to ($p_k$, $p_j$) is $p_t$ or the candidate blocker assigned to ($p_k$, $p_t$). But we have already shown that $p_i$ cannot block ($p_k$, $p_t$), so $p_i$ cannot block $(p_k,p_j)$. 

We will show that if $p_t$ blocks $(p_s,p_k)$ but not $(p_s,p_i)$, then $p_k$ cannot block $(p_i,p_s)$. See Figure \ref{fig:GroupTwo} (h). By Necessary Condition \ref{nc:blockedProperty} case 2, we have ($p_s$, $p_i$) is an invisible pair and $p_t$ is the candidate blocker assigned to it, a contradiction. 

This completes the cases when a ray intersects a segment, and we will now consider cases when no ray intersects a segment.  First suppose that there is a $p_v \in \p(p_{i+1},p_{s-1})$ so that $p_s$ blocks $(p_t,p_v)$ and $p_i$ blocks $(p_k,p_v)$.  It follows that the centers are $\{p_i,p_s\}$-pinched, and therefore they cannot be $\{p_k,p_t\}$-pinched by Necessary Condition \ref{nc:pinched}.  Therefore there cannot be a $p_u \in \p(p_{t+1},p_{k-1})$ such that $p_t$ blocks $(p_s,p_u)$ and $p_k$ blocks $(p_i,p_u)$.   See Figure \ref{fig:GroupTwo}.  Next we will show that a ray from one center cannot intersect both rays from the other center.  If it does, then we can assume without loss of generality that $p_s$ blocks $(p_t,p_i)$ and $(p_t,p_k)$, and that there is a point $p_j \in \p(p_{t+1},p_{k-1})$ such that $p_t$ is the designated blocker for invisible pair ($p_s$, $p_j$).  If $r_{k,i}$ intersects both rays of $L_{s,t}$, then $p_i$ should block $(p_k,p_t)$ but not ($p_k$, $p_j$).   By Necessary Condition \ref{nc:blockedProperty} case 2, we have that ($p_k$, $p_s$) is an invisible pair and $p_i$ is the candidate blocker assigned to ($p_k$, $p_s$). If $p_t$ is the candidate blocker assigned to the invisible pair ($p_s$, $p_j$), then ($p_k$, $p_j$) is an invisible pair and is assigned the candidate blocker $p_i$, a contradiction. See Figure \ref{fig:GroupTwo} (j).

\begin{figure}
\centering
\begin{tabular}{c@{\hspace{0.1\linewidth}}c}
\includegraphics[scale=0.4]{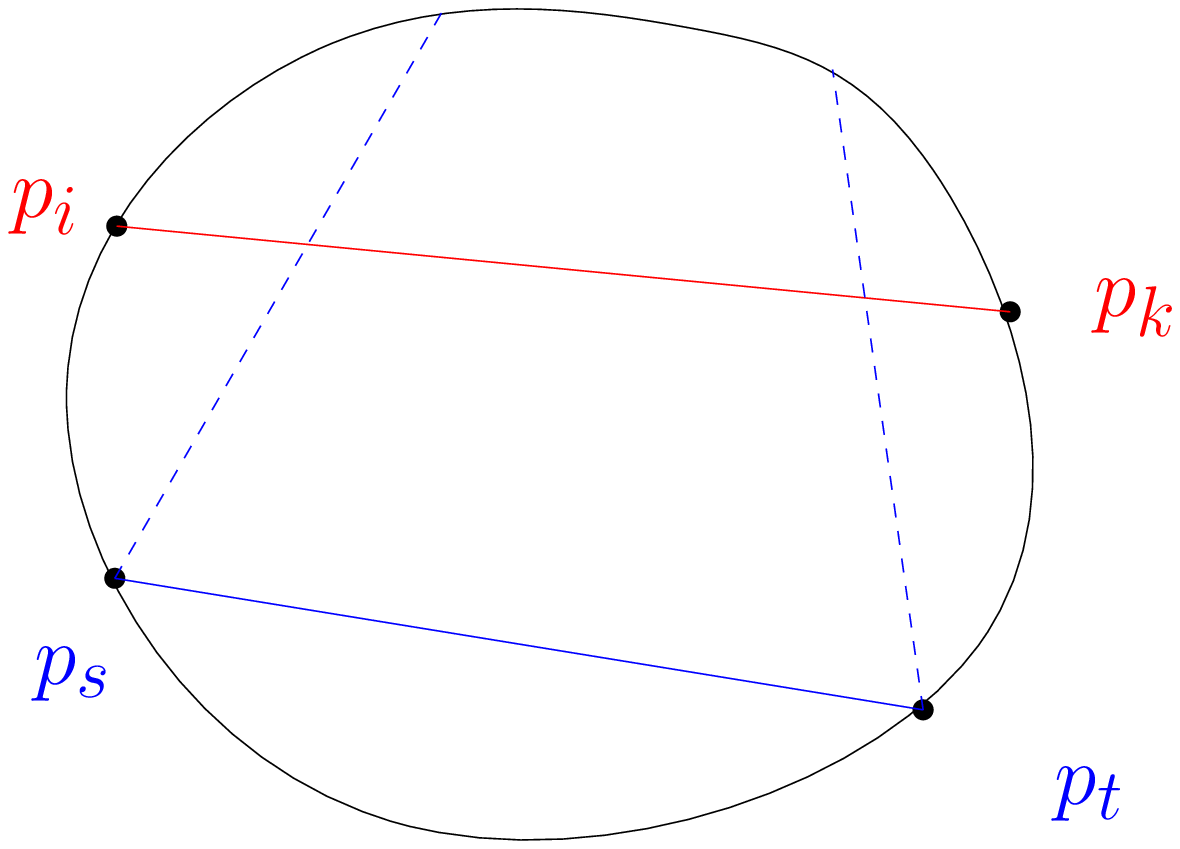} &
\includegraphics[scale=0.4]{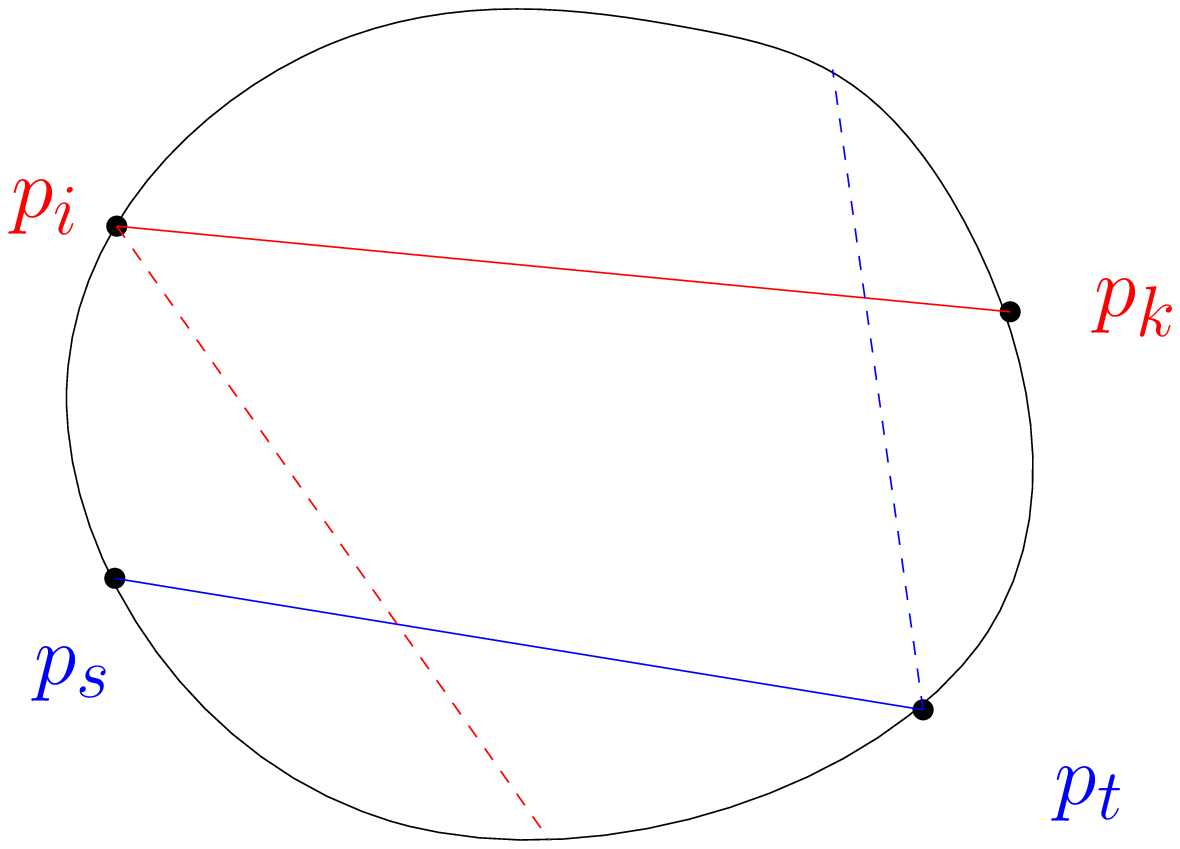} \\
(a)  & (b)   \\ 
\includegraphics[scale=0.4]{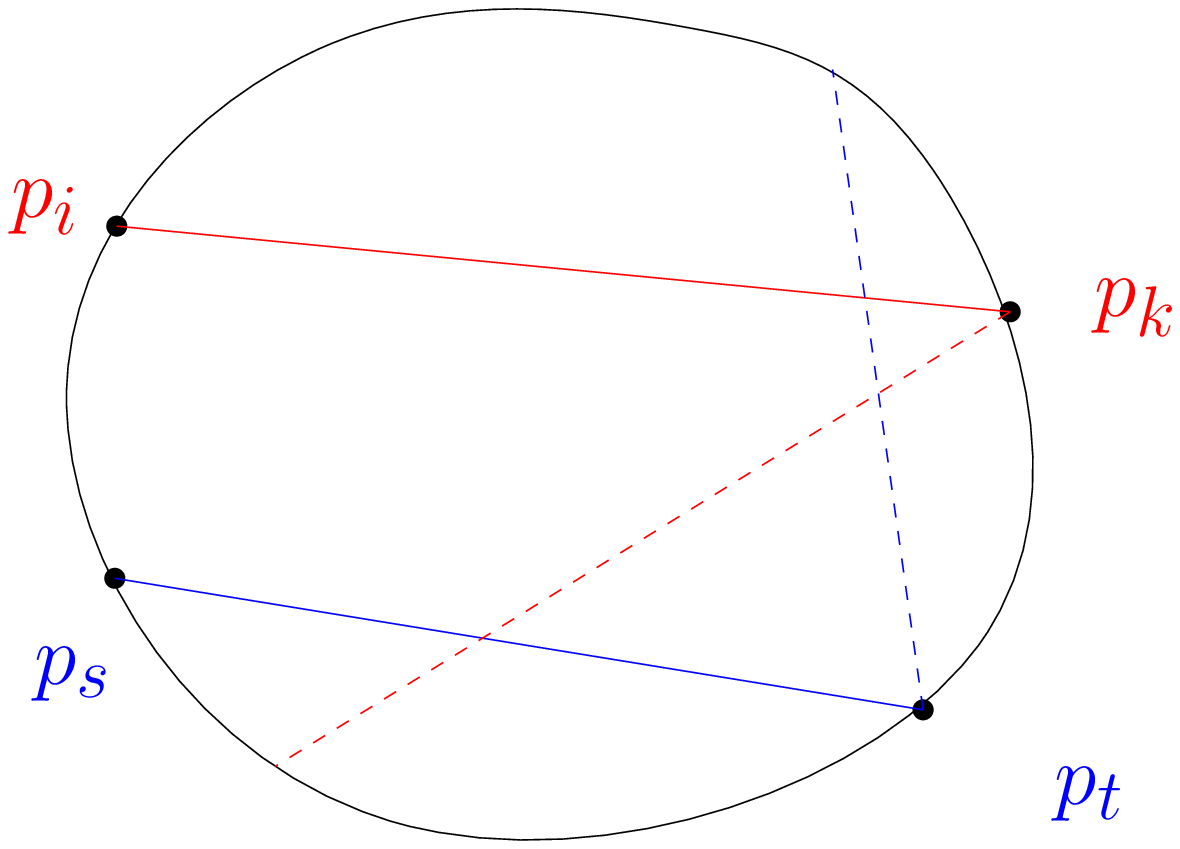} &
\includegraphics[scale=0.4]{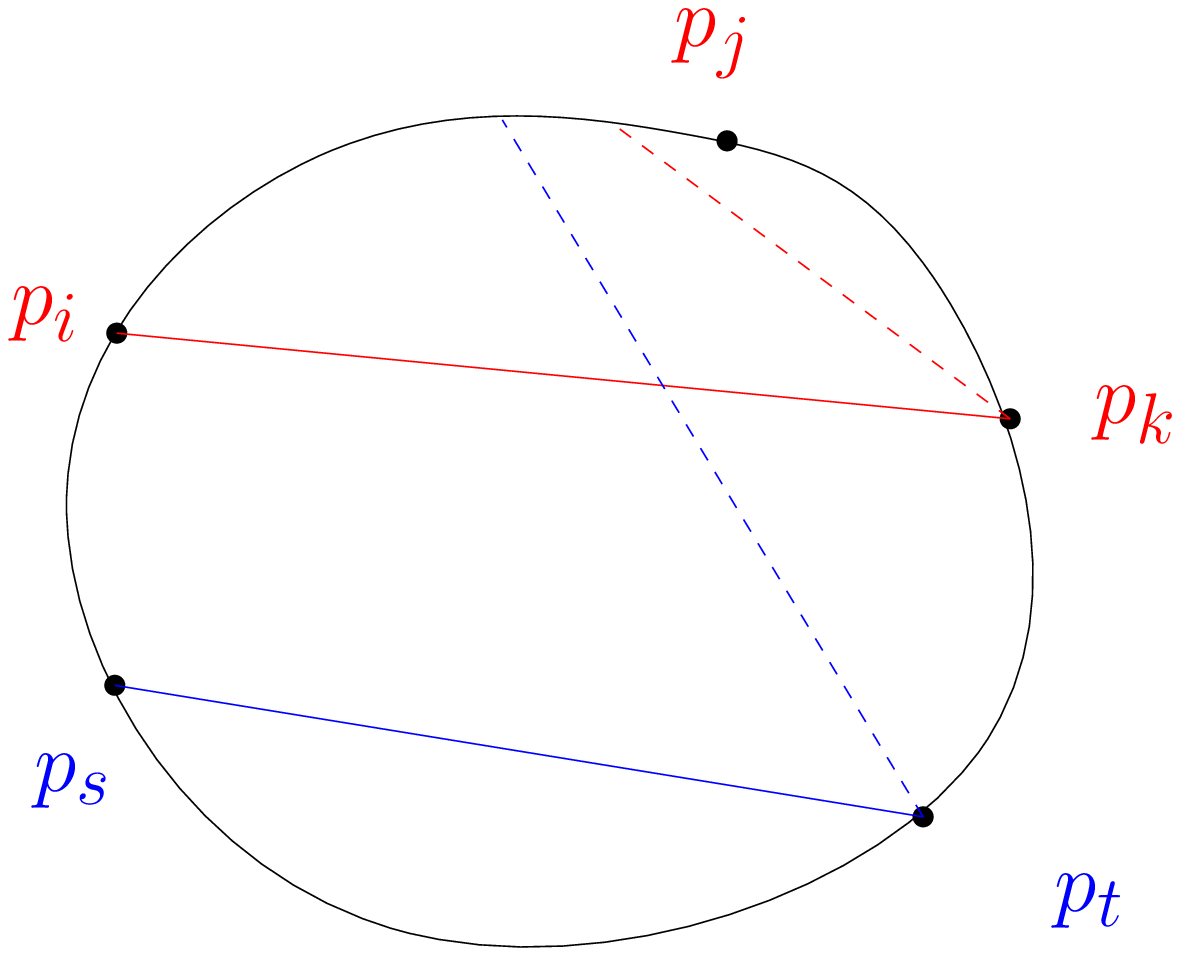} \\
(c)  & (d)    \\ 
\includegraphics[scale=0.4]{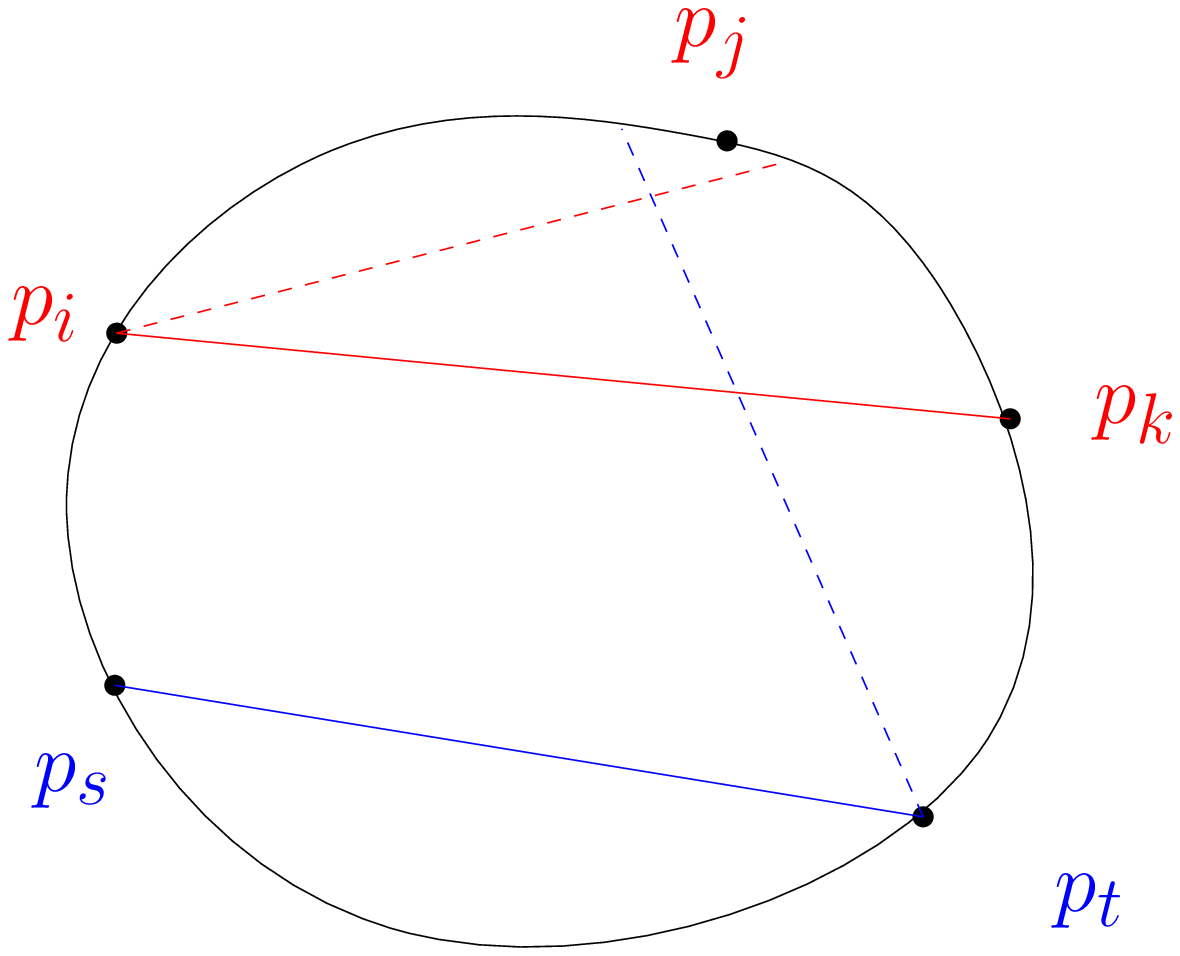} &
\includegraphics[scale=0.4]{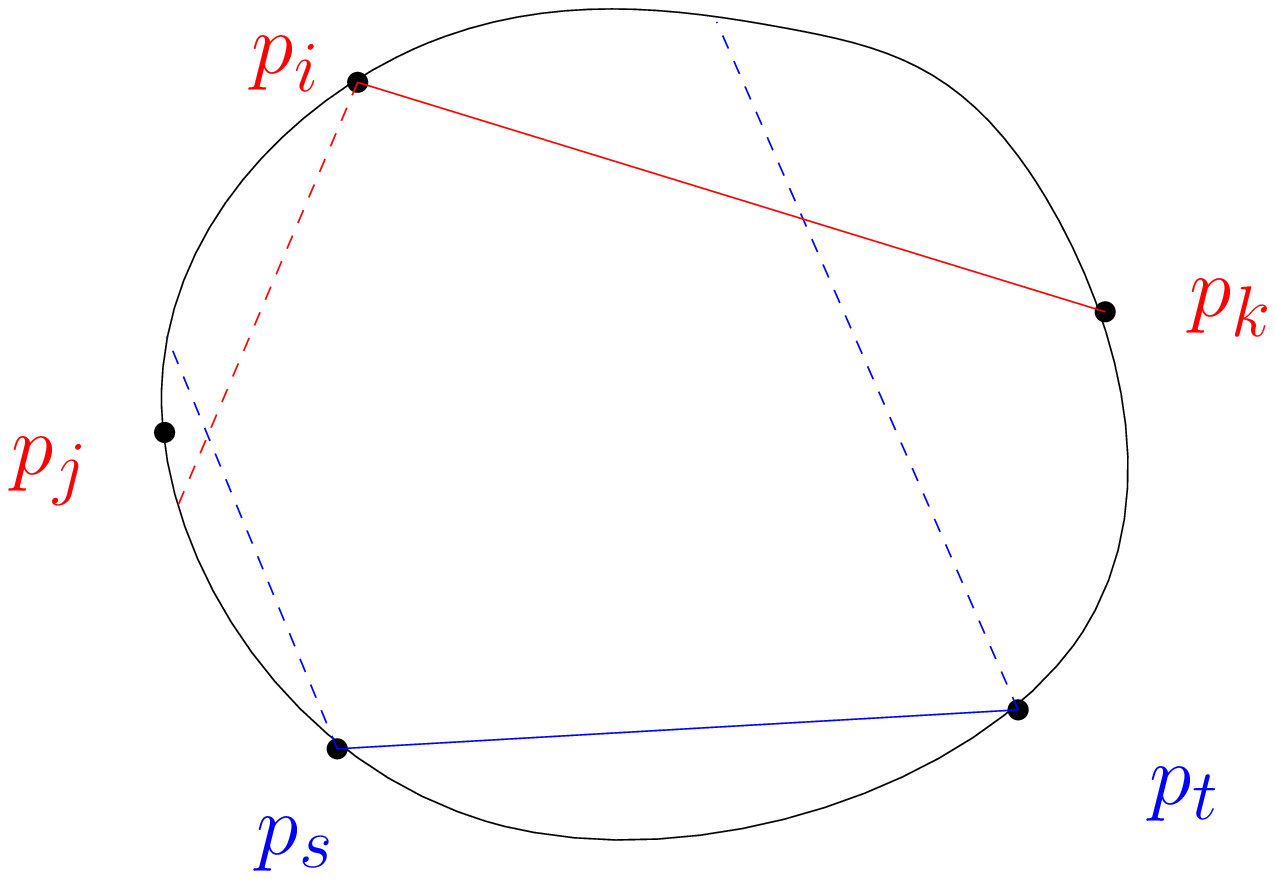} \\
(e)  & (f)    \\ 
\includegraphics[scale=0.4]{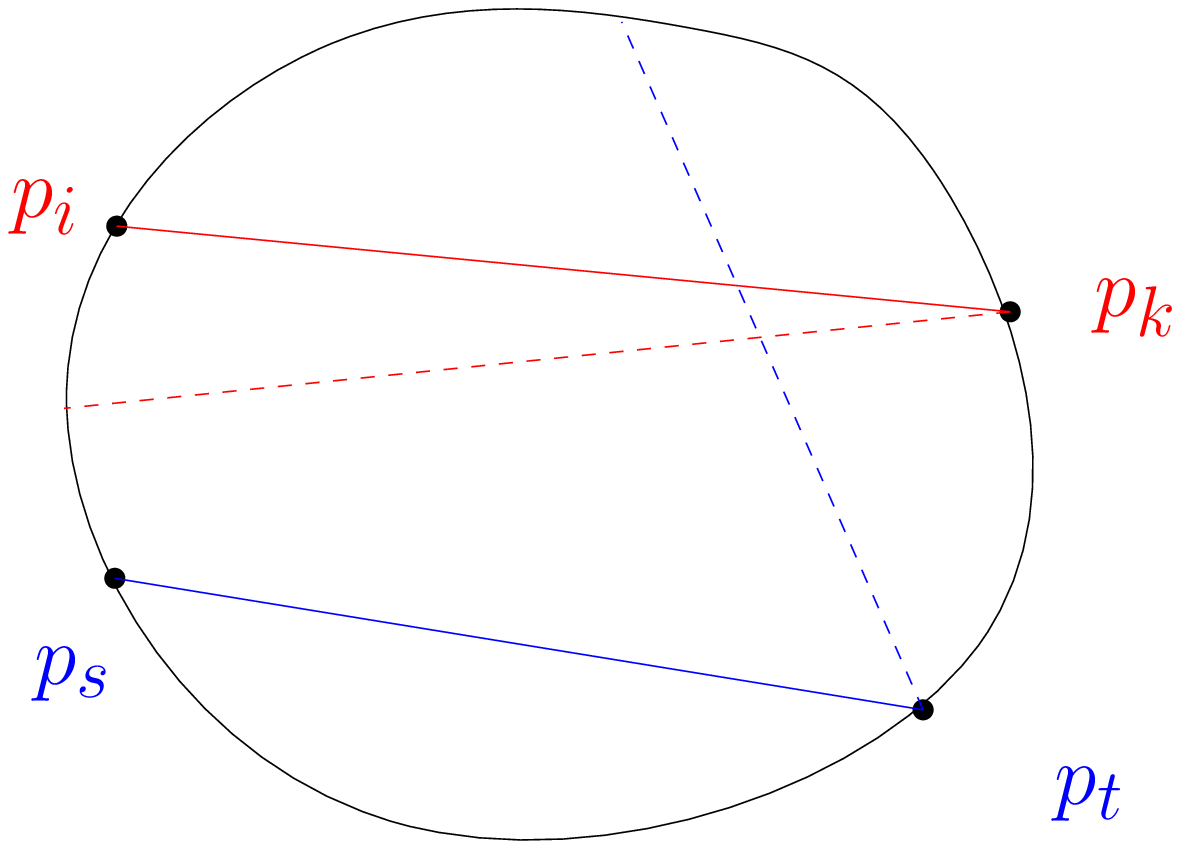} &
\includegraphics[scale=0.4]{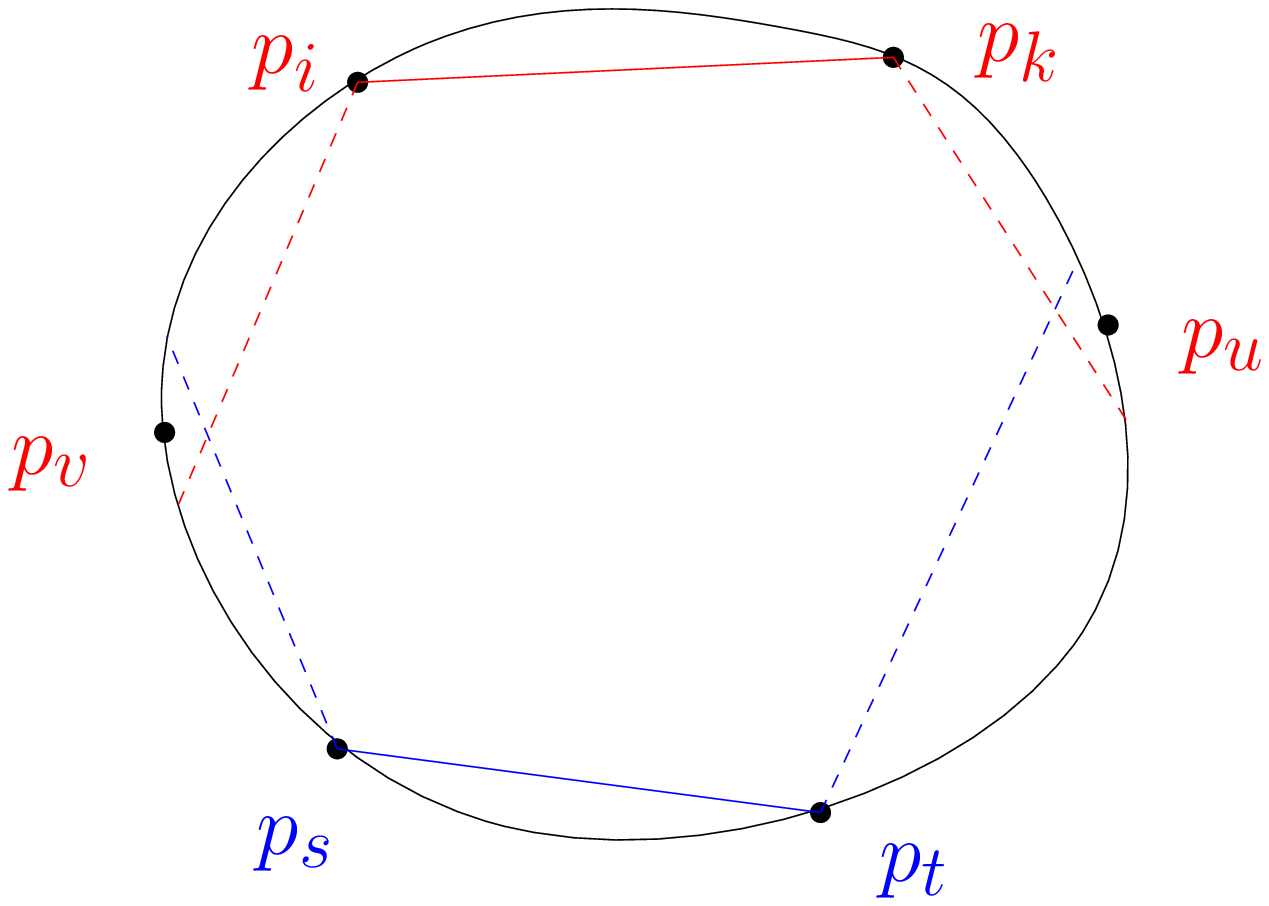} \\
(g)  & (h)    \\ 
\includegraphics[scale=0.4]{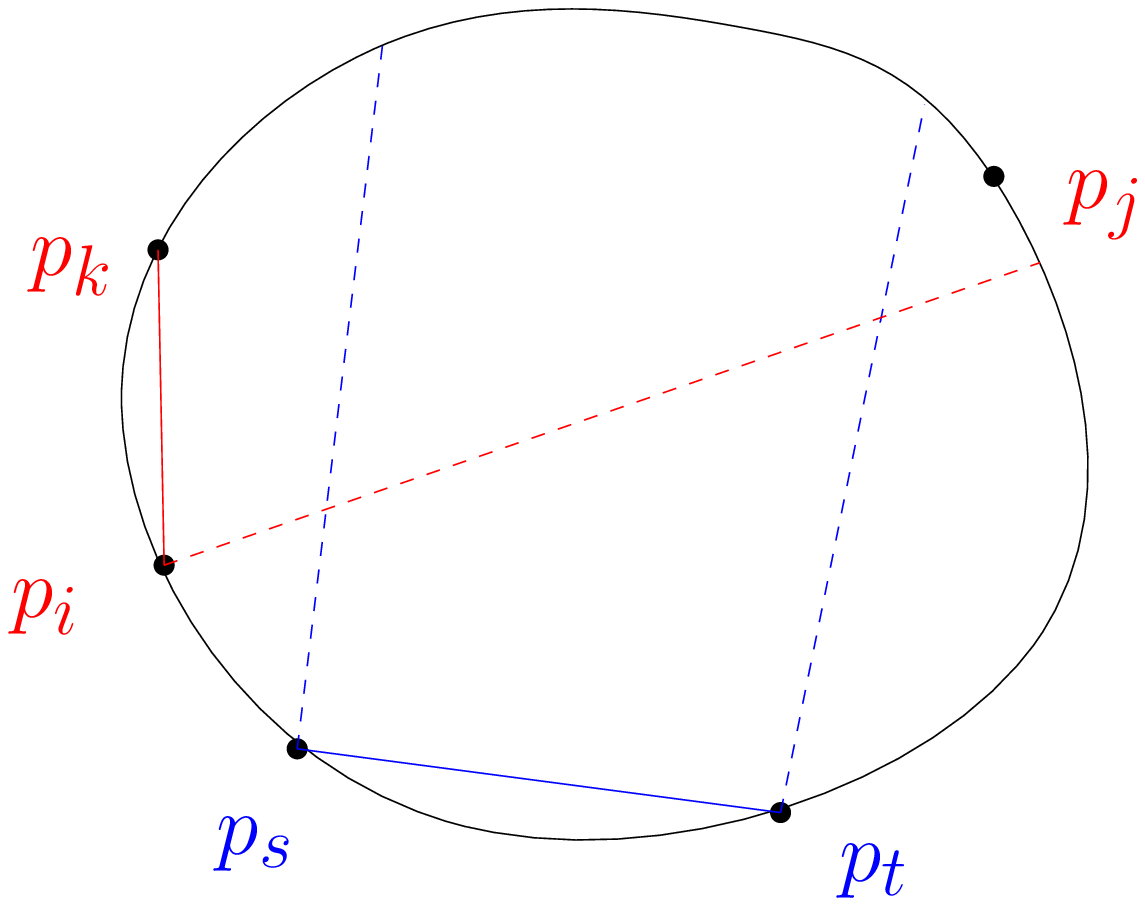}\\
(i)     \\ 
\end{tabular}
\caption{An illustration of segments do not intersect each other.}
\label{fig:GroupTwo}
\end{figure}

\textbf{Case 3:} Now suppose $p_k = p_t$ and therefore the centers share a vertex. Without loss of generality, assume $p_k \in \p(p_i,p_s)$. We will prove that the centers will not intersect anywhere else.  First note that a ray cannot intersect a segment in this situation, as we would contradict the definition of a candidate blocker. See Figure \ref{fig:GroupThree} (a). Now suppose $p_s$ blocks $(p_k,p_j)$ for some $p_j \in \p(p_{s+1}, p_{i-1})$.  First $p_i$ cannot block $p_k$ from seeing $p_j$, by Necessary Condition \ref{nc:blockers}, as otherwise $(p_k,p_j)$ would have two designated blockers. See Figure \ref{fig:GroupThree} (b).  Second, suppose $p_k$ blocks $(p_i,p_s)$, but does not block $(p_i,p_j)$.  By Necessary Condition \ref{nc:blockedProperty} case 1, we have that the candidate blocker assigned to the invisible pair ($p_s$, $p_i$) is $p_k$. But if $p_k$ is the candidate blocker assigned to the invisible pair ($p_s$, $p_i$), then ($p_i$, $p_j$) is an invisible pair and is assigned the candidate blocker $p_k$ by Necessary Condition \ref{nc:blockedProperty} case 1, a contradiction. See Figure \ref{fig:GroupThree} (c). Now consider $p_j \in \p(p_{i+1}, p_{k-1})$. First we will show that if $p_k$ blocks $(p_i,p_j)$, then $p_k$ also blocks $p_s$ from seeing $p_j$, by Necessary Condition \ref{nc:alreadyBlocked} case 1 we have that {$p_s$, $p_k$} is a visible pair, and the candidate blocker assigned to the invisible pair ($p_s$, $p_j$) is $p_k$. See Figure \ref{fig:GroupThree} (d). Second we will show that if $p_i$ blocks $(p_k,p_j)$, then $p_k$ also blocks $p_s$ from seeing $p_i$.  By Necessary Condition \ref{nc:alreadyBlocked}, we have that the candidate blocker for ($p_s$, $p_j$) is either $p_i$ or the candidate blocker assigned to ($p_s$, $p_i$), but in this case $p_k$ is not the candidate blocker assigned to ($p_s$, $p_i$), a contradiction. See Figure \ref{fig:GroupThree} (e).

\begin{figure}
\centering
\begin{tabular}{c@{\hspace{0.1\linewidth}}c}
\includegraphics[scale=0.4]{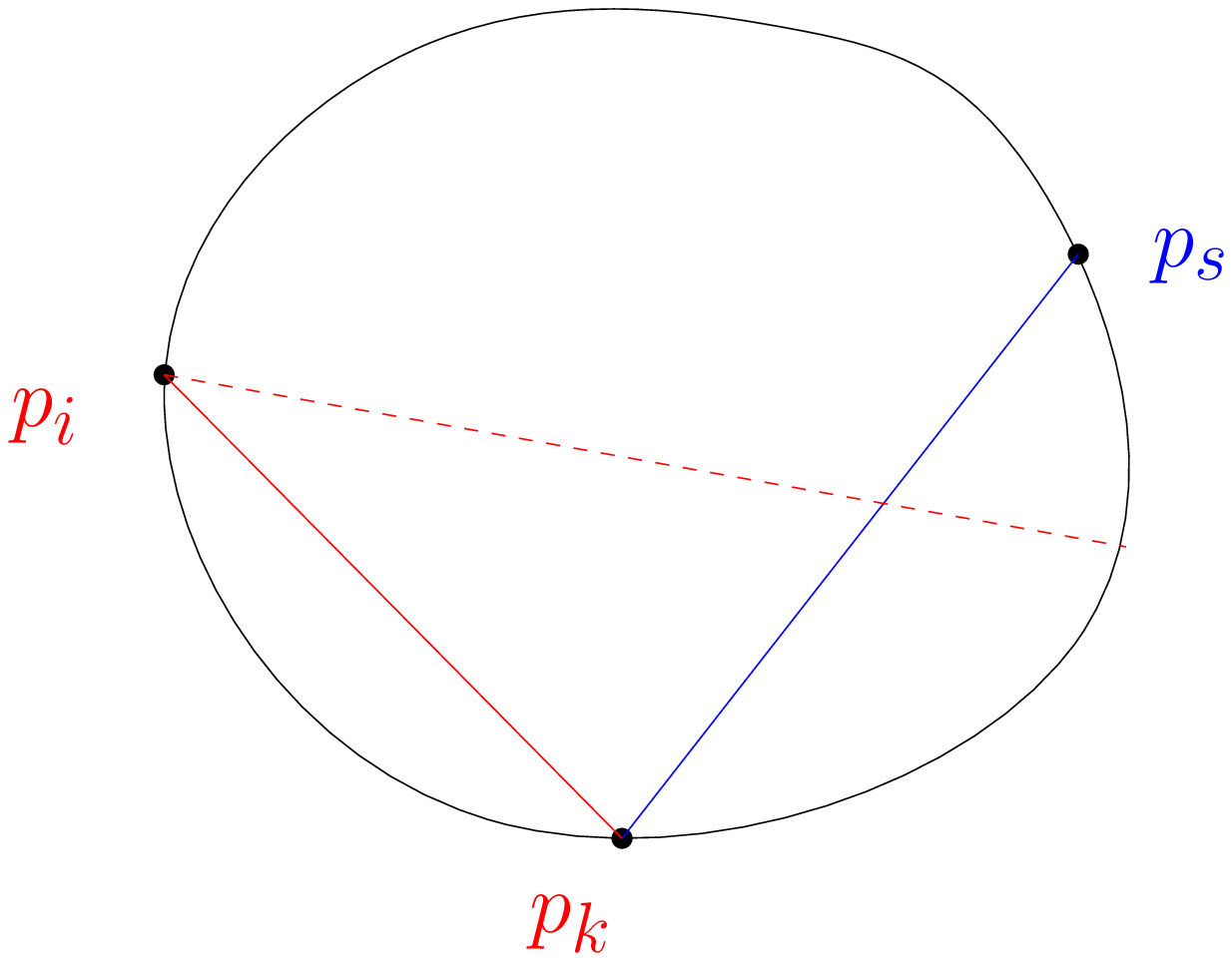} &
\includegraphics[scale=0.4]{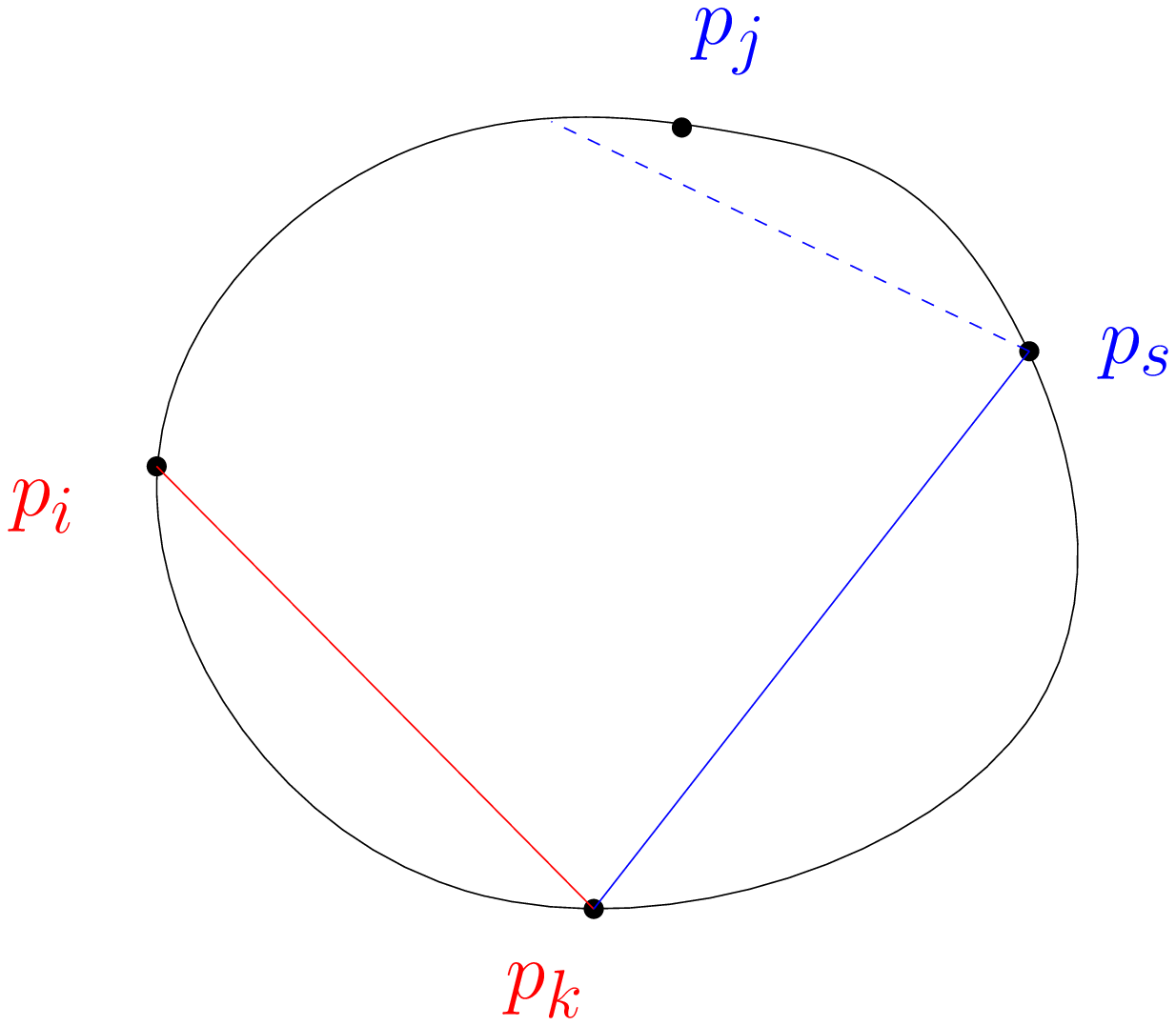} \\
(a)  & (b)   \vspace{1cm} \\ 
\includegraphics[scale=0.4]{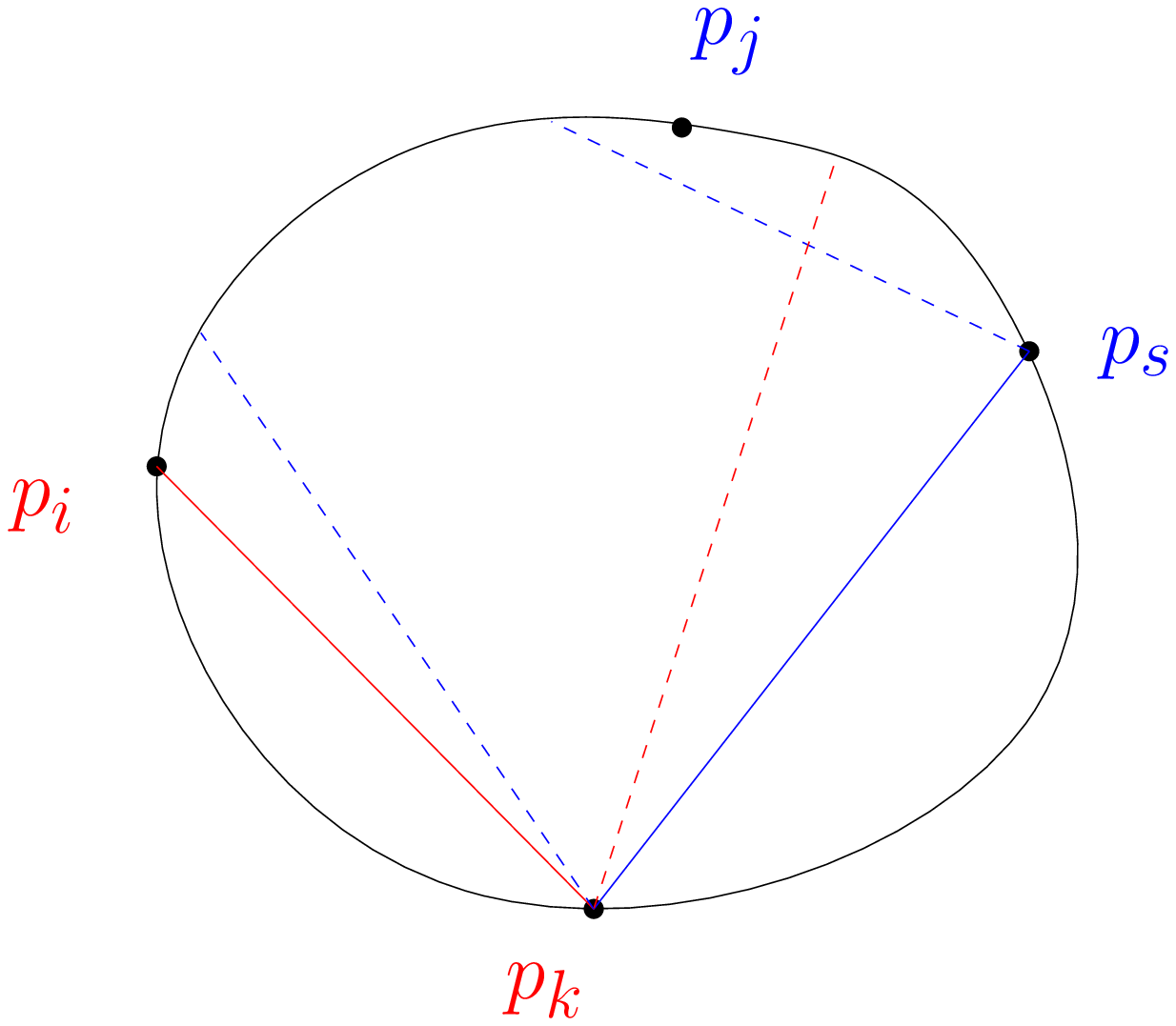} &
\includegraphics[scale=0.4]{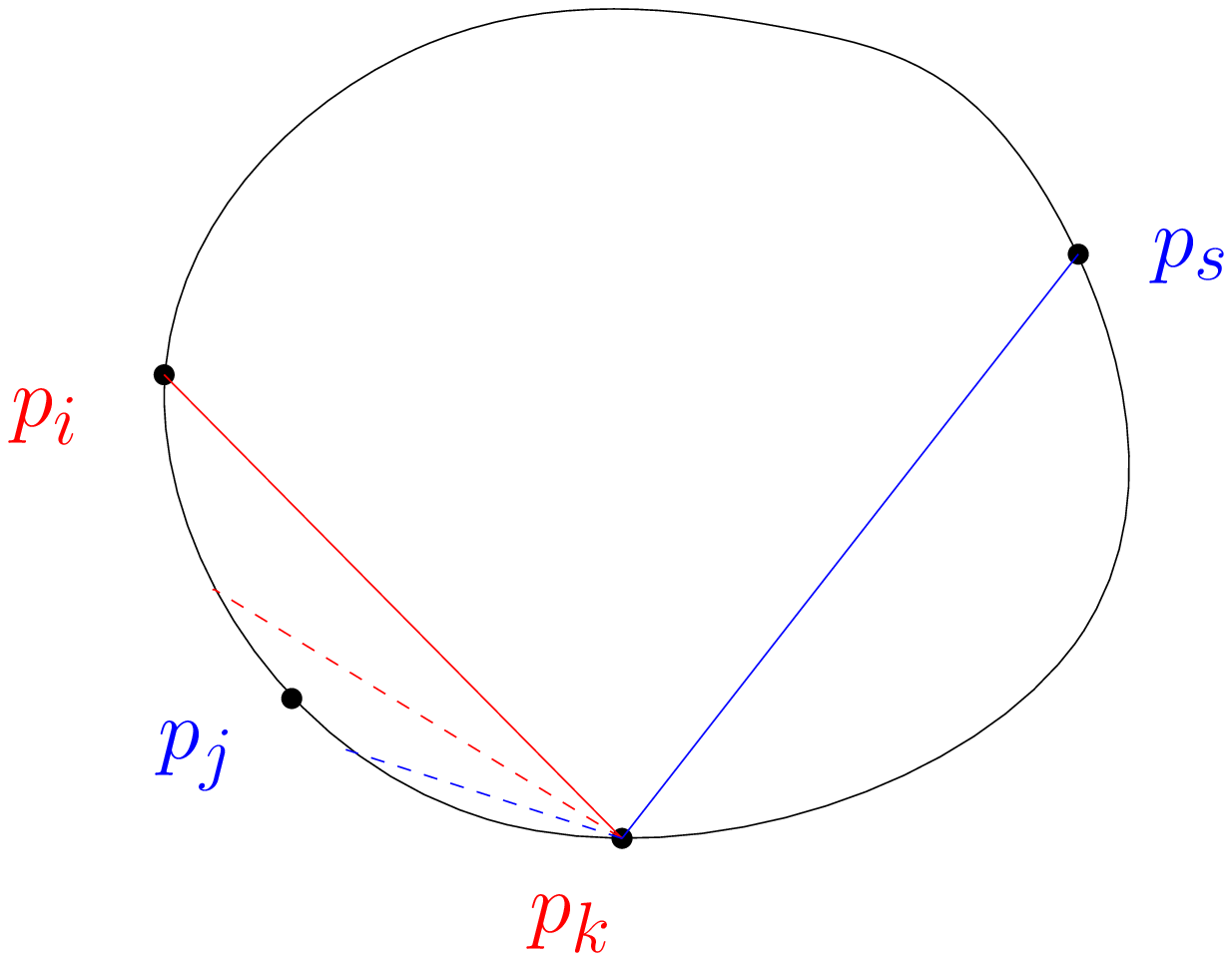} \\
(c)  & (d)   \vspace{1cm} \\ 
\includegraphics[scale=0.4]{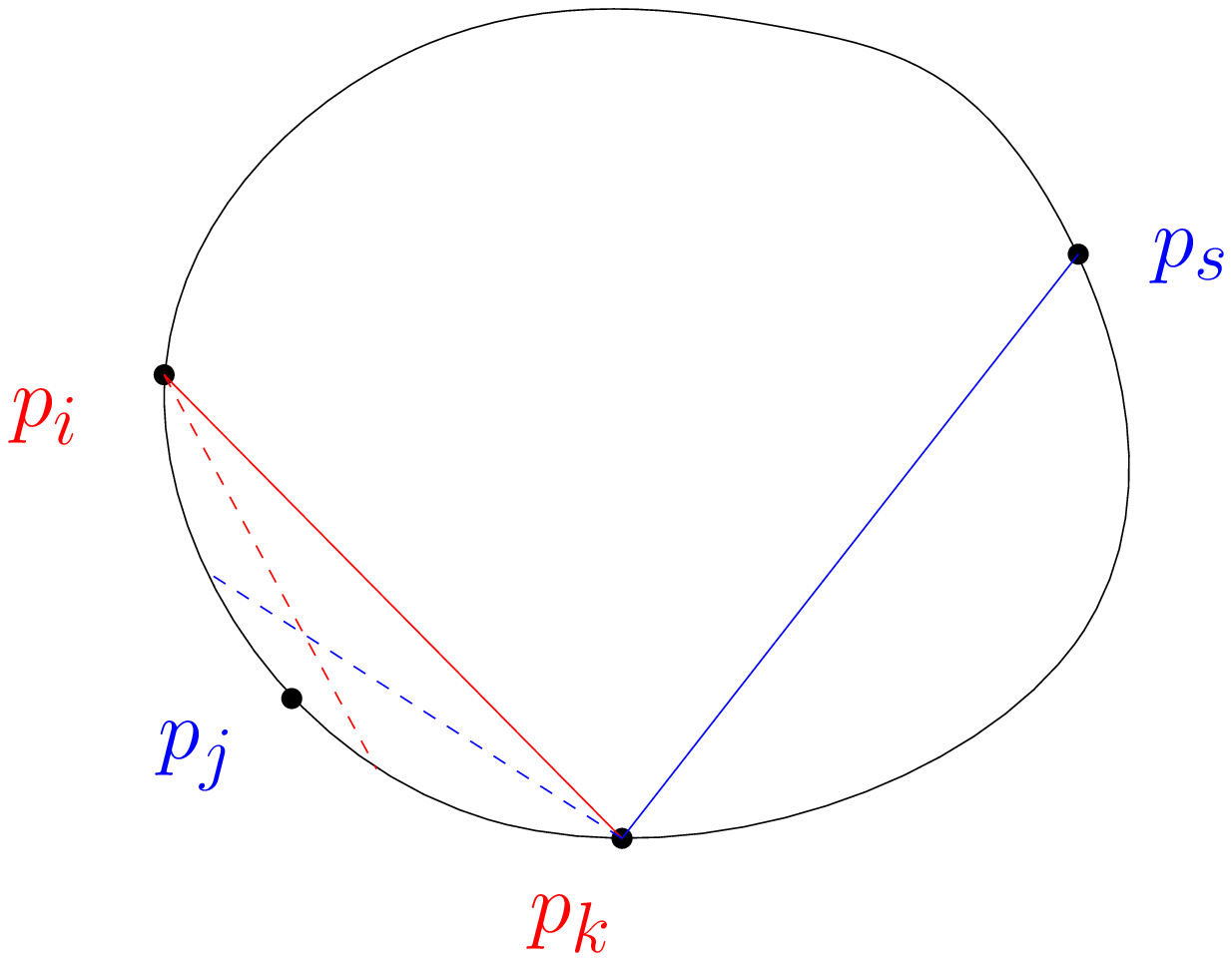} \\
(e)  \vspace{1cm} \\ 
\end{tabular}
\caption{An illustration of segments shared same endpoint.}
\label{fig:GroupThree}
\end{figure}
 \end{proof}

We now have that $G_{VE}$ is the vertex-edge visibility graph for some pseudo-polygon $P$.  It follows from Lemma \ref{lem:edgeVis} that $G$ is the visibility graph of $P$, giving us the following theorem.

\begin{theorem}
A graph $G$ with a given Hamiltonian cycle $C$ is the visibility graph of a pseudo-polygon $P$ if and only if there is an assignment of candidate blockers to the invisible pairs that satisfies Necessary Conditions 1 - 5.
\label{thm:characterize}
\end{theorem}

\appendix

\end{document}

%% file: nc1a.pstex_t
\begin{picture}(0,0)%
\includegraphics{nc1a.pstex}%
\end{picture}%
\setlength{\unitlength}{2901sp}%
\begingroup\makeatletter\ifx\SetFigFont\undefined%
\gdef\SetFigFont#1#2#3#4#5{%
  \reset@font\fontsize{#1}{#2pt}%
  \fontfamily{#3}\fontseries{#4}\fontshape{#5}%
  \selectfont}%
\fi\endgroup%
\begin{picture}(4350,3143)(2281,-4094)
\put(2296,-2401){\makebox(0,0)[lb]{\smash{{\SetFigFont{8}{9.6}{\rmdefault}{\mddefault}{\updefault}{\color[rgb]{0,0,0}$p_i$}%
}}}}
\put(4186,-4021){\makebox(0,0)[lb]{\smash{{\SetFigFont{8}{9.6}{\rmdefault}{\mddefault}{\updefault}{\color[rgb]{0,0,0}$p_k$}%
}}}}
\put(5716,-3661){\makebox(0,0)[lb]{\smash{{\SetFigFont{8}{9.6}{\rmdefault}{\mddefault}{\updefault}{\color[rgb]{0,0,0}$p_t$}%
}}}}
\put(6616,-2041){\makebox(0,0)[lb]{\smash{{\SetFigFont{8}{9.6}{\rmdefault}{\mddefault}{\updefault}{\color[rgb]{0,0,0}$p_j$}%
}}}}
\end{picture}%

%% file: nc1b.pstex_t
\begin{picture}(0,0)%
\includegraphics{nc1b.pstex}%
\end{picture}%
\setlength{\unitlength}{2901sp}%
\begingroup\makeatletter\ifx\SetFigFont\undefined%
\gdef\SetFigFont#1#2#3#4#5{%
  \reset@font\fontsize{#1}{#2pt}%
  \fontfamily{#3}\fontseries{#4}\fontshape{#5}%
  \selectfont}%
\fi\endgroup%
\begin{picture}(4350,3143)(2281,-4094)
\put(2296,-2401){\makebox(0,0)[lb]{\smash{{\SetFigFont{8}{9.6}{\rmdefault}{\mddefault}{\updefault}{\color[rgb]{0,0,0}$p_i$}%
}}}}
\put(4186,-4021){\makebox(0,0)[lb]{\smash{{\SetFigFont{8}{9.6}{\rmdefault}{\mddefault}{\updefault}{\color[rgb]{0,0,0}$p_k$}%
}}}}
\put(6616,-2041){\makebox(0,0)[lb]{\smash{{\SetFigFont{8}{9.6}{\rmdefault}{\mddefault}{\updefault}{\color[rgb]{0,0,0}$p_j$}%
}}}}
\end{picture}%

%% file: nc2a.pstex_t
\begin{picture}(0,0)%
\includegraphics{nc2a.pstex}%
\end{picture}%
\setlength{\unitlength}{2901sp}%
\begingroup\makeatletter\ifx\SetFigFont\undefined%
\gdef\SetFigFont#1#2#3#4#5{%
  \reset@font\fontsize{#1}{#2pt}%
  \fontfamily{#3}\fontseries{#4}\fontshape{#5}%
  \selectfont}%
\fi\endgroup%
\begin{picture}(4350,3143)(2281,-4094)
\put(2296,-2401){\makebox(0,0)[lb]{\smash{{\SetFigFont{8}{9.6}{\rmdefault}{\mddefault}{\updefault}{\color[rgb]{0,0,0}$p_i$}%
}}}}
\put(4186,-4021){\makebox(0,0)[lb]{\smash{{\SetFigFont{8}{9.6}{\rmdefault}{\mddefault}{\updefault}{\color[rgb]{0,0,0}$p_k$}%
}}}}
\put(6616,-2041){\makebox(0,0)[lb]{\smash{{\SetFigFont{8}{9.6}{\rmdefault}{\mddefault}{\updefault}{\color[rgb]{0,0,0}$p_j$}%
}}}}
\put(2836,-3436){\makebox(0,0)[lb]{\smash{{\SetFigFont{8}{9.6}{\rmdefault}{\mddefault}{\updefault}{\color[rgb]{0,0,0}$p_s$}%
}}}}
\put(2926,-1186){\makebox(0,0)[lb]{\smash{{\SetFigFont{8}{9.6}{\rmdefault}{\mddefault}{\updefault}{\color[rgb]{0,0,0}$p_t$}%
}}}}
\end{picture}%

%% file: nc2b.pstex_t
\begin{picture}(0,0)%
\includegraphics{nc2b.pstex}%
\end{picture}%
\setlength{\unitlength}{2901sp}%
\begingroup\makeatletter\ifx\SetFigFont\undefined%
\gdef\SetFigFont#1#2#3#4#5{%
  \reset@font\fontsize{#1}{#2pt}%
  \fontfamily{#3}\fontseries{#4}\fontshape{#5}%
  \selectfont}%
\fi\endgroup%
\begin{picture}(4350,3143)(2281,-4094)
\put(2296,-2401){\makebox(0,0)[lb]{\smash{{\SetFigFont{8}{9.6}{\rmdefault}{\mddefault}{\updefault}{\color[rgb]{0,0,0}$p_i$}%
}}}}
\put(4186,-4021){\makebox(0,0)[lb]{\smash{{\SetFigFont{8}{9.6}{\rmdefault}{\mddefault}{\updefault}{\color[rgb]{0,0,0}$p_k$}%
}}}}
\put(6616,-2041){\makebox(0,0)[lb]{\smash{{\SetFigFont{8}{9.6}{\rmdefault}{\mddefault}{\updefault}{\color[rgb]{0,0,0}$p_j$}%
}}}}
\put(2836,-3436){\makebox(0,0)[lb]{\smash{{\SetFigFont{8}{9.6}{\rmdefault}{\mddefault}{\updefault}{\color[rgb]{0,0,0}$p_s$}%
}}}}
\put(3511,-3751){\makebox(0,0)[lb]{\smash{{\SetFigFont{8}{9.6}{\rmdefault}{\mddefault}{\updefault}{\color[rgb]{0,0,0}$p_t$}%
}}}}
\end{picture}%

%% file: nc3a.pstex_t
\begin{picture}(0,0)%
\includegraphics{nc3a.pstex}%
\end{picture}%
\setlength{\unitlength}{2901sp}%
\begingroup\makeatletter\ifx\SetFigFont\undefined%
\gdef\SetFigFont#1#2#3#4#5{%
  \reset@font\fontsize{#1}{#2pt}%
  \fontfamily{#3}\fontseries{#4}\fontshape{#5}%
  \selectfont}%
\fi\endgroup%
\begin{picture}(4350,3439)(2281,-4094)
\put(2296,-2401){\makebox(0,0)[lb]{\smash{{\SetFigFont{8}{9.6}{\rmdefault}{\mddefault}{\updefault}{\color[rgb]{0,0,0}$p_i$}%
}}}}
\put(4186,-4021){\makebox(0,0)[lb]{\smash{{\SetFigFont{8}{9.6}{\rmdefault}{\mddefault}{\updefault}{\color[rgb]{0,0,0}$p_k$}%
}}}}
\put(6616,-2041){\makebox(0,0)[lb]{\smash{{\SetFigFont{8}{9.6}{\rmdefault}{\mddefault}{\updefault}{\color[rgb]{0,0,0}$p_j$}%
}}}}
\put(2836,-3436){\makebox(0,0)[lb]{\smash{{\SetFigFont{8}{9.6}{\rmdefault}{\mddefault}{\updefault}{\color[rgb]{0,0,0}$p_s$}%
}}}}
\put(3916,-826){\makebox(0,0)[lb]{\smash{{\SetFigFont{8}{9.6}{\rmdefault}{\mddefault}{\updefault}{\color[rgb]{0,0,0}$p_a$}%
}}}}
\end{picture}%

%% file: nc3b.pstex_t
\begin{picture}(0,0)%
\includegraphics{nc3b.pstex}%
\end{picture}%
\setlength{\unitlength}{2901sp}%
\begingroup\makeatletter\ifx\SetFigFont\undefined%
\gdef\SetFigFont#1#2#3#4#5{%
  \reset@font\fontsize{#1}{#2pt}%
  \fontfamily{#3}\fontseries{#4}\fontshape{#5}%
  \selectfont}%
\fi\endgroup%
\begin{picture}(4350,3484)(2281,-4094)
\put(2296,-2401){\makebox(0,0)[lb]{\smash{{\SetFigFont{8}{9.6}{\rmdefault}{\mddefault}{\updefault}{\color[rgb]{0,0,0}$p_i$}%
}}}}
\put(4186,-4021){\makebox(0,0)[lb]{\smash{{\SetFigFont{8}{9.6}{\rmdefault}{\mddefault}{\updefault}{\color[rgb]{0,0,0}$p_k$}%
}}}}
\put(6616,-2041){\makebox(0,0)[lb]{\smash{{\SetFigFont{8}{9.6}{\rmdefault}{\mddefault}{\updefault}{\color[rgb]{0,0,0}$p_j$}%
}}}}
\put(3511,-961){\makebox(0,0)[lb]{\smash{{\SetFigFont{8}{9.6}{\rmdefault}{\mddefault}{\updefault}{\color[rgb]{0,0,0}$p_t$}%
}}}}
\put(5041,-781){\makebox(0,0)[lb]{\smash{{\SetFigFont{8}{9.6}{\rmdefault}{\mddefault}{\updefault}{\color[rgb]{0,0,0}$p_t$}%
}}}}
\end{picture}%

%% file: nc3c.pstex_t
\begin{picture}(0,0)%
\includegraphics{nc3c.pstex}%
\end{picture}%
\setlength{\unitlength}{1865sp}%
\begingroup\makeatletter\ifx\SetFigFont\undefined%
\gdef\SetFigFont#1#2#3#4#5{%
  \reset@font\fontsize{#1}{#2pt}%
  \fontfamily{#3}\fontseries{#4}\fontshape{#5}%
  \selectfont}%
\fi\endgroup%
\begin{picture}(4350,3484)(2281,-4094)
\put(2296,-2401){\makebox(0,0)[lb]{\smash{{\SetFigFont{5}{6.0}{\rmdefault}{\mddefault}{\updefault}{\color[rgb]{0,0,0}$p_i$}%
}}}}
\put(4186,-4021){\makebox(0,0)[lb]{\smash{{\SetFigFont{5}{6.0}{\rmdefault}{\mddefault}{\updefault}{\color[rgb]{0,0,0}$p_k$}%
}}}}
\put(6616,-2041){\makebox(0,0)[lb]{\smash{{\SetFigFont{5}{6.0}{\rmdefault}{\mddefault}{\updefault}{\color[rgb]{0,0,0}$p_j$}%
}}}}
\put(5041,-781){\makebox(0,0)[lb]{\smash{{\SetFigFont{5}{6.0}{\rmdefault}{\mddefault}{\updefault}{\color[rgb]{0,0,0}$p_q$}%
}}}}
\end{picture}%

%% file: nc3d.pstex_t
\begin{picture}(0,0)%
\includegraphics{nc3d.pstex}%
\end{picture}%
\setlength{\unitlength}{1865sp}%
\begingroup\makeatletter\ifx\SetFigFont\undefined%
\gdef\SetFigFont#1#2#3#4#5{%
  \reset@font\fontsize{#1}{#2pt}%
  \fontfamily{#3}\fontseries{#4}\fontshape{#5}%
  \selectfont}%
\fi\endgroup%
\begin{picture}(4350,3143)(2281,-4094)
\put(2296,-2401){\makebox(0,0)[lb]{\smash{{\SetFigFont{5}{6.0}{\rmdefault}{\mddefault}{\updefault}{\color[rgb]{0,0,0}$p_i$}%
}}}}
\put(4186,-4021){\makebox(0,0)[lb]{\smash{{\SetFigFont{5}{6.0}{\rmdefault}{\mddefault}{\updefault}{\color[rgb]{0,0,0}$p_k$}%
}}}}
\put(6616,-2041){\makebox(0,0)[lb]{\smash{{\SetFigFont{5}{6.0}{\rmdefault}{\mddefault}{\updefault}{\color[rgb]{0,0,0}$p_j$}%
}}}}
\put(5761,-3481){\makebox(0,0)[lb]{\smash{{\SetFigFont{5}{6.0}{\rmdefault}{\mddefault}{\updefault}{\color[rgb]{0,0,0}$p_q$}%
}}}}
\put(2836,-3436){\makebox(0,0)[lb]{\smash{{\SetFigFont{5}{6.0}{\rmdefault}{\mddefault}{\updefault}{\color[rgb]{0,0,0}$p_s$}%
}}}}
\end{picture}%

%% file: nc3e.pstex_t
\begin{picture}(0,0)%
\includegraphics{nc3e.pstex}%
\end{picture}%
\setlength{\unitlength}{1865sp}%
\begingroup\makeatletter\ifx\SetFigFont\undefined%
\gdef\SetFigFont#1#2#3#4#5{%
  \reset@font\fontsize{#1}{#2pt}%
  \fontfamily{#3}\fontseries{#4}\fontshape{#5}%
  \selectfont}%
\fi\endgroup%
\begin{picture}(4350,3259)(2281,-4094)
\put(2296,-2401){\makebox(0,0)[lb]{\smash{{\SetFigFont{5}{6.0}{\rmdefault}{\mddefault}{\updefault}{\color[rgb]{0,0,0}$p_i$}%
}}}}
\put(4186,-4021){\makebox(0,0)[lb]{\smash{{\SetFigFont{5}{6.0}{\rmdefault}{\mddefault}{\updefault}{\color[rgb]{0,0,0}$p_k$}%
}}}}
\put(6616,-2041){\makebox(0,0)[lb]{\smash{{\SetFigFont{5}{6.0}{\rmdefault}{\mddefault}{\updefault}{\color[rgb]{0,0,0}$p_j$}%
}}}}
\put(5761,-3481){\makebox(0,0)[lb]{\smash{{\SetFigFont{5}{6.0}{\rmdefault}{\mddefault}{\updefault}{\color[rgb]{0,0,0}$p_q$}%
}}}}
\put(3196,-1006){\makebox(0,0)[lb]{\smash{{\SetFigFont{5}{6.0}{\rmdefault}{\mddefault}{\updefault}{\color[rgb]{0,0,0}$p_t$}%
}}}}
\end{picture}%